\documentclass[a4paper,accepted=2023-07-05]{quantumarticle}
\pdfoutput=1

\usepackage{graphicx, color, graphpap}% Include figure files
\usepackage{float}
\usepackage{enumitem}
\usepackage{amssymb}
\usepackage{amsfonts}
\usepackage{amsthm}
\usepackage[numbers,sort&compress]{natbib}
\usepackage{multirow}
\usepackage[dvipsnames]{xcolor}
\usepackage{hyperref}
\usepackage[T1]{fontenc}
\usepackage{bbm}
\usepackage{thmtools,thm-restate}
\usepackage{verbatim}
\usepackage{mathtools}
\usepackage{titlesec}
\usepackage{dsfont}
\usepackage{diagbox}
\usepackage[caption = false]{subfig}
\usepackage{subfig}
\usepackage{tikz}
\usetikzlibrary{quantikz}
\usepackage{titlesec}
\usepackage{appendix}

\long\def\ca#1\cb{} %Use for commenting out: \ca...\cb

\newcommand{\ceil}[1]{\left\lceil{#1}\right\rceil}
\newcommand{\floor}[1]{\left\lfloor{#1}\right\rfloor}

\newcommand{\bramatket}[3]{\langle #1 \hspace{1pt} | #2 | \hspace{1pt} #3 \rangle}

\newcommand{\dya}[1]{\ket{#1}\!\bra{#1}}
\newcommand{\BC}{\mathcal{B}}

\newcommand{\DC}{\mathcal{D}}

\newcommand{\GC}{\mathcal{G}}

\newcommand{\IC}{\mathcal{I}}

\newcommand{\MC}{\mathcal{M}}

\newcommand{\OC}{\mathcal{O}}

\newcommand{\TC}{\mathcal{T}}
\newcommand{\UC}{\mathcal{U}}
\newcommand{\VC}{\mathcal{V}}
\newcommand{\WC}{\mathcal{W}}

\newcommand{\YC}{\mathcal{Y}}

\newcommand{\Tr}{{\rm Tr}}

\renewcommand{\geq}{\geqslant}
\renewcommand{\leq}{\leqslant}

\renewcommand{\vec}[1]{\boldsymbol{#1}}  % Bold vectors instead of arrow vectors

\newcommand*{\id}{\openone}
\newcommand{\thv}{\vec{\theta}}

{}%         Body font
{}%         Indent amount (empty = no indent, \parindent 

\newtheorem{lemma}{Lemma}
\newtheorem*{remark}{Remark}
\newtheorem{supproposition}{Supplemental Proposition}
\newtheorem{corollary}{Corollary}
\newtheorem{proposition}{Proposition}

\newtheorem{definition}{Definition}

\makeatletter
\newcommand{\dbloverline}[1]{\overline{\dbl@overline{#1}}}
\newcommand{\dbl@overline}[1]{\mathpalette\dbl@@overline{#1}}
\newcommand{\dbl@@overline}[2]{%
  \begingroup
  \sbox\z@{$\m@th#1\overline{#2}$}%
  \ht\z@=\dimexpr\ht\z@-2\dbl@adjust{#1}\relax
  \box\z@
  \ifx#1\scriptstyle\kern-\scriptspace\else
  \ifx#1\scriptscriptstyle\kern-\scriptspace\fi\fi
  \endgroup
}
\newcommand{\dbl@adjust}[1]{%
  \fontdimen8
  \ifx#1\displaystyle\textfont\else
  \ifx#1\textstyle\textfont\else
  \ifx#1\scriptstyle\scriptfont\else
  \scriptscriptfont\fi\fi\fi 3
}
\makeatother

\usepackage[usestackEOL]{stackengine}[2013-10-15]
\def\x{\hspace{3ex}}
\def\y{\hspace{2.45ex}}
\def\z{\hspace{1.9ex}}
\stackMath

\begin{document}

\title{The battle of clean and dirty qubits in the era of partial error correction}

\author{Daniel Bultrini$^*$}
\affiliation{Theoretical Division, Los Alamos National Laboratory, Los Alamos, NM 87545, USA}
\affiliation{Theoretische Chemie, Physikalisch-Chemisches Institut, Universität Heidelberg, INF 229, D-69120 Heidelberg, Germany}

\author{Samson Wang$^*$}
\affiliation{Theoretical Division, Los Alamos National Laboratory, Los Alamos, NM 87545, USA}
\affiliation{Imperial College London, London, UK}

\author{Piotr Czarnik}
\affiliation{Theoretical Division, Los Alamos National Laboratory, Los Alamos, NM 87545, USA}
\affiliation{Institute of Theoretical Physics, Jagiellonian University, Krakow, Poland.}

\author{Max Hunter Gordon}
\affiliation{Theoretical Division, Los Alamos National Laboratory, Los Alamos, NM 87545, USA}
\affiliation{Instituto de Física Teórica, UAM/CSIC, Universidad Autónoma de Madrid, Madrid 28049, Spain}

\author{M. Cerezo}
\affiliation{Information Sciences, Los Alamos National Laboratory, Los Alamos, NM 87545, USA}
\affiliation{Quantum Science Center, Oak Ridge, TN 37931, USA}

\author{Patrick J. Coles}
\affiliation{Theoretical Division, Los Alamos National Laboratory, Los Alamos, NM 87545, USA}
\affiliation{Quantum Science Center, Oak Ridge, TN 37931, USA}

\author{Lukasz Cincio}
\affiliation{Theoretical Division, Los Alamos National Laboratory, Los Alamos, NM 87545, USA}
\affiliation{Quantum Science Center, Oak Ridge, TN 37931, USA}

\begin{abstract}
When error correction becomes possible it will be necessary to dedicate a large number of physical qubits to each logical qubit. Error correction allows for deeper circuits to be run, but each additional physical qubit can potentially contribute an exponential increase in computational space, so there is a trade-off between using qubits for error correction or using them as noisy qubits. In this work, we look at the effects of using noisy qubits in conjunction with noiseless qubits (an idealized model for error-corrected qubits), which we call the ``clean and dirty'' setup. We employ analytical models and numerical simulations to characterize this setup. Numerically we show the appearance of Noise-Induced Barren Plateaus (NIBPs), i.e., an exponential concentration of observables caused by noise, in an Ising model Hamiltonian variational ansatz circuit. We observe this even if only a single qubit is noisy and given a deep enough circuit, suggesting that NIBPs cannot be fully overcome simply by error-correcting a subset of the qubits. On the positive side, we find that for every noiseless qubit in the circuit, there is an exponential suppression in the concentration of gradient observables, showing the benefit of partial error correction. Finally, our analytical models corroborate these findings by showing that observables concentrate with a scaling in the exponent related to the ratio of dirty-to-total qubits.
\def\thefootnote{*}\footnotetext{These authors contributed equally to this work.}\def\thefootnote{\arabic{footnote}}

\end{abstract}
\maketitle

\section{Introduction} 

Since Feynman's seminal proposal~\cite{feynman1982simulating}, quantum computing has moved from a purely theoretical exercise to real devices. Yet, quantum computers today are still far from an ideal machine able to truly harness the full power of fault-tolerant computing. Current devices do not have the qubit counts or gate fidelities required to implement large-scale error correction. This means that only small-scale implementations of error-correcting codes are currently feasible~\cite{egan2021fault}. Therefore, implementing Shor's factoring algorithm \cite{shor1994algorithms}, the Harrow-Hassidim-Lloyd algorithm~\cite{harrow2009quantum}, and the myriad of possibilities this computational paradigm could bring are still a distant reality.

An important primitive in quantum computing is estimating the expectation values of some operator at the end of a quantum circuit. This appears in the Noisy Intermediate Scale Quantum (NISQ)~\cite{preskill2018quantum} setting for the case of variational quantum algorithms~\cite{cerezo2020variationalreview,bharti2021noisy} and quantum machine learning~\cite{biamonte2017quantum},  but also in fault-tolerant algorithms~\cite{nielsen2000quantum}. However, it has been shown that noise can have detrimental effects resulting in the exponential concentration of expectation values \cite{aharonov1996limitations, ben2013quantum,franca2020limitations,wang2020noise}. In the specific case of variational quantum algorithms and quantum machine learning, this implies untrainability of the models, i.e., exponential scaling of resources required for training~\cite{wang2020noise,mcclean2018barren, cerezo2021cost, arrasmith2021equivalence, arrasmith2020effect, cerezo2020impact}.
This scaling phenomenon is known as a Noise-Induced Barren Plateau (NIBP)~\cite{wang2020noise}, which is a specific type of barren plateau ~\cite{mcclean2018barren, cerezo2021cost,arrasmith2021equivalence, arrasmith2020effect, cerezo2020impact,marrero2020entanglement, larocca2021diagnosing,holmes2021connecting, thanasilp2021subtleties}. In addition, it was also shown that error mitigation cannot reverse the effect of noise on expectation value concentration~\cite{wang2021can}. Therefore, noise presents one of the largest obstacles to the practical utility of quantum computers and to obtaining near-term quantum advantage. This leads to a critical question that we aim to address in this work: \textit{how can we mitigate, or remove, the effect of NIBPs and exponential concentration due to noise?}

To investigate this question we construct a quantum computational model that employs ``clean'' and ``dirty'' qubits. Here, we define clean qubits as qubits on which no noise channels are applied, which we envision as the best-case implementation of an error-correcting code. On the other hand, the dirty qubits in our setup are affected  by noise. This setup presents a novel, physically motivated paradigm where we can explore the scalability of cost concentration in quantum algorithms. We can motivate this setup as an idealized scenario of the regime where the number of qubits becomes sufficiently large to implement error correction on some, but not all, of the qubits in a quantum computer. The intersection of NISQ and quantum error correction has been considered in other recent works~\cite{cao2021nisq, holmes2020nisq+, suzuki2022quantum}. We note our setup is different from earlier clean qubit models \cite{knill1998power,fujii2016power,morimae2017power}, which work on a system where one or more qubits are in a known state, such as $\ket{0}$, and the rest are in the maximally mixed state~\footnote{Such models have the ability to factor numbers \cite{gidney2017factoring}, compute partition functions \cite{chowdhury2021computing} and are hard to simulate classically \cite{fujii2018impossibility}.}.

We find that using clean qubits can mitigate the effect of NIBPs. Specifically, we find that using clean qubits increases the depth accessible to an ansatz before the magnitude of the gradients vanishes. However, NIBPs are not avoided unless all qubits are clean. We analytically explore the effect of clean qubits on the concentration of expectation values and gradient using two toy models. The first is a CNOT ladder acting on an input state followed by depolarizing noise on the dirty qubits. The second explores two initially disjoint subsystems, one or both of which are noisy, that are subsequently connected via CNOTs. Both show that the expectation value concentrates with a scaling in the exponent related to the ratio of dirty to total qubits.

Our numerics arrive at essentially the same conclusions as our analytics. Here, we numerically explore the magnitude of the gradient when optimizing the Hamiltonian Variational Ansatz (HVA) for systems of $4,\ 6,$ and $8$ qubits under a depolarizing noise model and also a more realistic model of a trapped ion quantum computer. This allows us to investigate how the gradient scales with system size under the effect of noise. We show that using clean qubits leads to gradient scaling similar to the case of  reducing  error rates of the noise on every qubit. This in turn means that deeper circuits can be employed before NIBPs become an issue. Therefore, our work presents a possible avenue to improve the reach of  quantum computing using a combination of error-corrected and dirty qubits.

\section{The Clean and Dirty Computational Model} \label{sec:clean_dirty}

In this work, we investigate how the emergence of exponential concentration and NIBPs is affected by introducing clean, noiseless  qubits in a quantum computer. Such a setup may provide insights into an early era of quantum error correction when, due to limited numbers of physical qubits available, the number of error-corrected, high-quality logical qubits will be severely restricted~\cite{laflamme1996perfect,gottesman2009introduction,fowler2012surface}. In such a case, a  question arises as to whether combining a limited number of high-quality logical qubits  with uncorrected, noisy qubits could mitigate or perhaps even prevent the exponential concentration  of expectation values and gradients. If so, this would highlight the power of post-NISQ quantum computing architectures, relative to their NISQ counterparts.

\begin{figure}[t]
    \centering
    \begin{quantikz}
       \lstick[wires=3]{$\underset{\textrm{Dirty qubits}}{n_\textrm{d}}$}    &\gate[wires=1]{U}&\gate[wires=1,style={rounded corners}]{\mathcal N}& \gate[wires=1,style={}]{U}   & \gate[wires=1,style={rounded corners}]{\mathcal {N}} & \qw\\
                                                                    &\vdots           &                                                    & \gate[nwires=1]{U}\qwbundle[alternate=2]{}& \gate[nwires=1,style={rounded corners}]{\mathcal {N}}\qwbundle[alternate=2]{}&\vdots \\
                                                                    &\gate[wires=1]{U}&\gate[wires=1,style={rounded corners}]{\mathcal N}& \gate[wires=2]{U}                       & \gate[wires=1,style={rounded corners}]{\mathcal {N}} & \qw\\
       \lstick[wires=3]{$\underset{\textrm{Clean qubits}}{n_\textrm{c}}$}    &\gate[wires=1]{U}& \qw                                              & \qw                                     & \qw                                                    & \qw\\
                                                                    &\vdots           &                                                     & \gate[nwires=1]{U}\qwbundle[alternate=2]{}& \qwbundle[alternate=2]{}&\vdots\\
                                                                    &\gate[wires=1]{U}& \qw                                              & \gate[wires=1,style={}]{U}   & \qw                                                    & \qw
    \end{quantikz}
    \caption{\textbf{The clean and dirty setup.} A schematic example of the clean and dirty setup.  We have  $n_c$ noiseless, clean qubits at which  no errors happen. For these qubits, we represent the action of a gate by its unitary $U$. We also have $n_d$ uncorrected, dirty qubits. The action of a gate at these qubits is represented by $U$ followed by a noise channel $\mathcal{N}$.  In the case of gates acting at both the dirty and the clean qubits, we represent them by $U$ followed by a noise channel $\mathcal{N}$ acting  only at the dirty qubits.}\label{fig:NCsetup}
\end{figure}
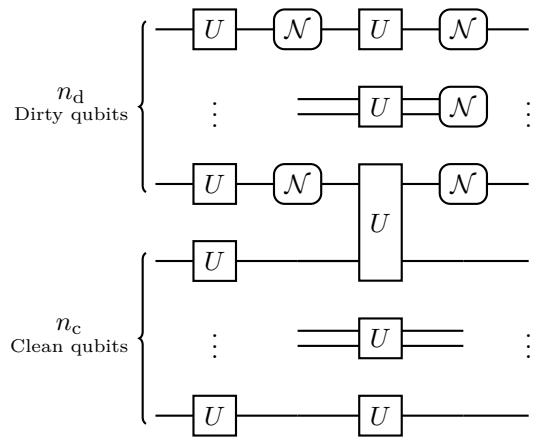

To provide insight into these questions, we consider an idealized model of such a quantum computer, with $n$ qubits,  partitioned into $n_c$ noiseless, clean qubits, and $n_d$ noisy, dirty qubits. We assume that  the action of a gate $g$, defined by a unitary $U$, on the clean qubits of a state $\rho$ is
\begin{equation}
\rho \to U \rho U^{\dag}.
\label{eq:clean}
\end{equation}
When $g$ acts at the dirty qubits we represent its action by  $U$ followed by a noise channel 
$\mathcal{N}$ acting only at the dirty qubits:
\begin{equation}
\rho \to \mathcal{N}(U \rho U^{\dag}). 
\label{eq:noisy}
\end{equation}
For a gate $g$ acting between a dirty and clean qubit, we assume no noise acts on the clean qubit. We call this setup  the  clean and dirty setup, and we illustrate it schematically in  Figure~\ref{fig:NCsetup}. This simple setup can be used to understand the fundamental  limitations  of a quantum computation involving both high-quality logical,  and low-quality, either  uncorrected or logical  qubits. For $n_d=0$ the setup is equivalent to a perfect, noiseless quantum computation, and for $n_c=0$ it becomes equivalent to a noisy quantum device.

We note that this idealized setup neglects error rates of real-world logical qubits due to imperfections of quantum error correction. Furthermore, the model does not account for the larger number of native hardware gates necessary to implement logical non-Clifford gates compared to the Clifford ones~\cite{kitaev1997quantum}. This overhead may result in a higher effective error rate for logical non-Cliffords that the clean-and-dirty setup does not consider.   
Therefore, we treat this work as an exploration of the potential of quantum algorithms combining qubits with large and small effective error rates rather than a realistic model of future quantum devices. 

An alternative approach for early implementations of quantum error correction is to use less robust quantum error correction codes with lower distances to detect and correct a part of the errors. An example of this approach is presented in Ref.~\cite{self2022protecting}.
 In general, the error rates of logical qubits are expected to decay exponentially with the code distance, with the base of the exponential being proportional to the ratio of the physical qubits' error rate to its threshold value~\cite{fowler2012surface}. Therefore, at the beginning of the quantum error correction era, when the physical error rates will be close to the threshold,  obtaining high-quality logical qubits would require a large code distance and consequently a large number of physical qubits per logical qubit~\cite{fowler2012surface}.
Consequently, implementing less powerful codes would require much fewer physical qubits making them more suitable for the first realizations of quantum error correction.

Nevertheless, to obtain a physical qubit count large enough for a quantum advantage, it might be necessary to supplement such logical qubits with noisy, physical ones. Furthermore, a partition of the device qubits into lower-quality error-corrected qubits and noisy ones would increase the number of computational qubits whilst potentially retaining some advantages of quantum error correction. Therefore, a combination of the noisy and quantum-error-corrected qubits modeled by the clean and dirty setup might be the optimal solution when a limited number of device qubits is the primary limitation on the computer's computational power. 

While this motivation applies to the first generations of quantum computers utilizing quantum error correction, the clean and dirty setup can also provide insights into later, more advanced architectures combining high and lower-quality logical qubits. Such a combination might be desirable to optimally utilize qubit counts of quantum hardware when the number of qubits is insufficient to use high-quality error corrections for all logical qubits. Furthermore, in the case of algorithms having qubits with particularly many gates acting upon them, much more errors will occur at these qubits. In such a case, using a code with a higher distance for these qubits might be a natural choice. An example of such an algorithm is quantum phase estimation with a large unitary controlled by an ancilla qubit~\cite{somma2020quantum}.
Therefore,  we expect insights gained with the clean and dirty model to potentially remain relevant beyond the early era of quantum error correction.

\section{Analytical Investigation}\label{sec:analytics}

\subsection{Concentration of expectation values}\label{sec:expval-conc}

In this section, we present two analytical models  in which exponential scaling effects due to noise are observed, even though a proportion of the registers are noise-free. Such scaling effects are present even if there is only one noisy register. The two settings demonstrate two different mechanisms in which noise can spread through a system via entangling gates. In the first setting, we consider a simple circuit where clean and dirty qubits interact via cycles of CNOT gates. Here, we find that the CNOTs allow entropy to spread from noisy to clean registers despite the fact that if the initial state is in the computational basis, no entanglement is actually created between the subsystems. In the second case, we consider a setting where two subsystems, which can in principle be initially disconnected, undergo an entangled measurement. We show that, due to the entangled measurement, local noise in the dirty qubits leads to decoherence and information loss in the joint system of all clean and dirty qubits. In both models, we find that the output distribution of the circuit, as measured in the computational basis, is exponentially indistinguishable from the uniform distribution with increasing circuit depth. However, in both settings, this concentration effect is dampened with increasing number of clean qubits, compared to the number of dirty qubits.

In the first setting, we consider the interaction of local noise with CNOT ladders, as shown in Figure \ref{fig:ladder}. Ladders of 2-qubit gates can be found as a primitive quantum machine learning settings (e.g. in \cite{havlivcek2019supervised}) as well as in variational quantum algorithms settings (e.g.  for the unitary coupled cluster ansatz when one implements an exponentiation of a multi-qubit Pauli operator \cite{taube2006new}). In order to understand how such ladders interact with noise, we consider a simplified model where layers of CNOT ladders are interleaved with instances of local depolarizing noise, where the local noise only acts on the first $n_d$ qubits. In the following proposition, proved in  Appendix~\ref{app:th}, we show how the output distribution converges with circuit depth.

\begin{figure}[t]
    \begin{center}
	\includegraphics[width= .7 \columnwidth]{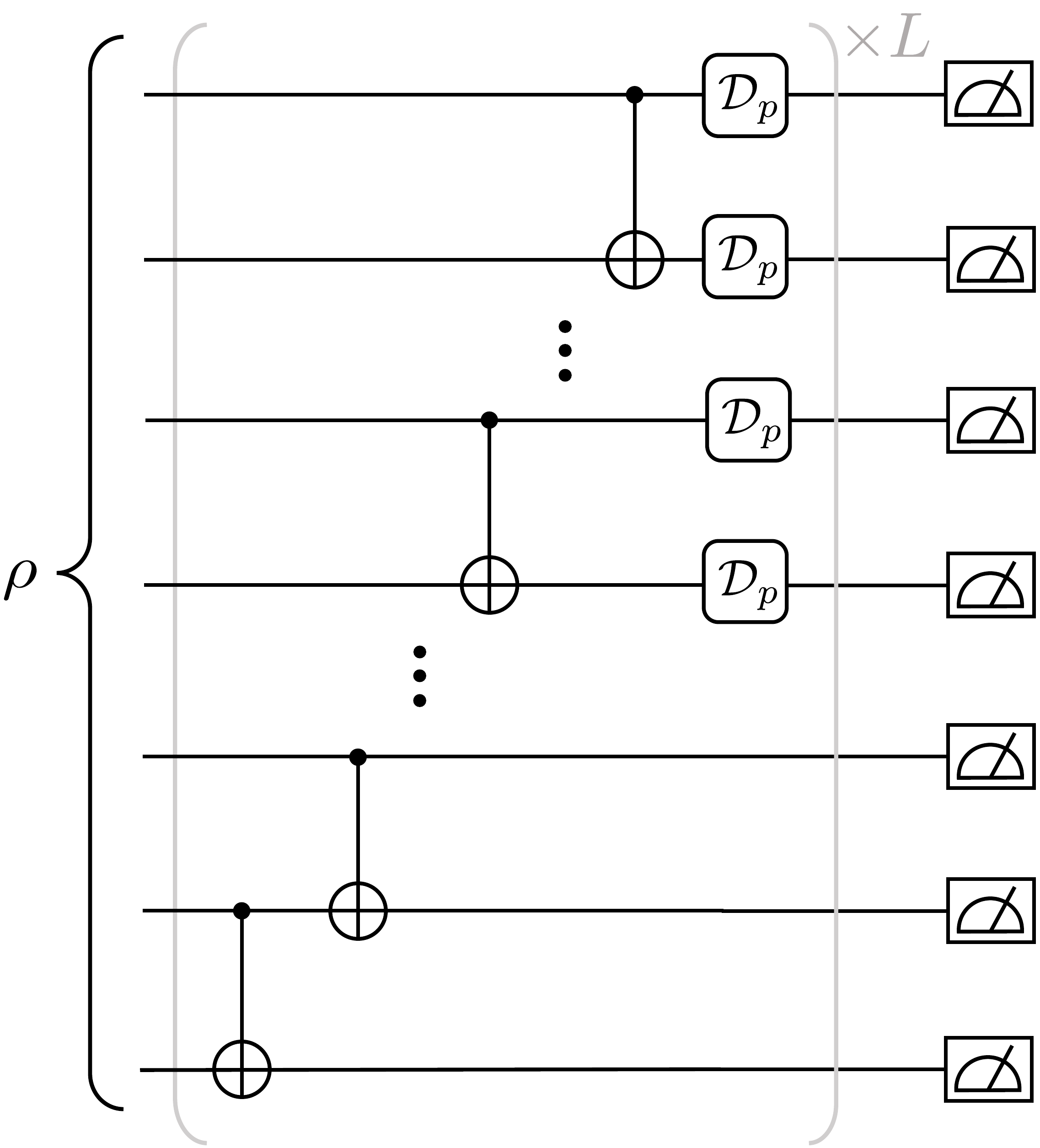}
    \end{center}
	\caption{\textbf{CNOT ladders.} In our first analytical example we consider an input state $\rho$ acted on by $L$ layers of CNOT ladders. After each layer, we consider an instance of local depolarizing noise $\DC_p$ with depolarizing probability $p$ acting on each of the first $n_d$ registers. We compare the output distribution as measured in the computational basis to the uniform distribution.
	}\label{fig:ladder}
\end{figure}

\begin{proposition}\label{prop:conc-ladder-cost}
Consider the circuit in Figure \ref{fig:ladder} with $L$ layers, $n_d$ dirty qubits each with depolarizing noise with depolarizing probability $p$ in each layer, and computational basis measurement $O_{z_n}=\dya{z_n}$; $z_n \in \{0,1\}^n$. Denote the channel that describes the action of the CNOT ladders and noise in the circuit as $\WC^{(n_d)}_L$. For any input state $\rho$, the expectation value concentrates as
\begin{equation}\label{eq:conc-ladder-cost}
    \Tr[\WC^{(n_d)}_L(\rho)O_{z_n}] - \frac{1}{2^n} \leq  (1-p)^{\Gamma_{n,n_d,L}}\sqrt{P(\rho)}\,,
\end{equation}
for any computational basis measurement, where $P(\rho)=\Tr[\rho^2]$ denotes the purity of $\rho$,
and we have defined
\begin{align}
    \Gamma_{n,n_d,L} =
    \begin{cases}
    \left\lfloor L\cdot \frac{n_d }{2^{\floor{\log_2 n}}} \right\rfloor\,, & \text{for}\ n_d = 1 \,,\\
    \left\lfloor L\cdot \frac{n_d }{2^{\ceil{\log_2 n}}} \right\rfloor\,, & \text{for}\ n_d > 1\,.
    \end{cases} 
\end{align}
\end{proposition}

Proposition \ref{prop:conc-ladder-cost} shows that the output distribution of the CNOT ladder circuit exponentially concentrates in the number of layers of the circuit. Moreover, compared to the setting where all registers are noisy, the exponent is scaled by a factor approximately equal to $n_d/n$. Thus, we can consider the scaling to depend on an \textit{effective} noisy depth $L_{\textrm{eff}}:= Ln_d/n$. We note that the bound in Eq.~\eqref{eq:conc-ladder-cost} holds even if there is no entanglement in the circuit - that is, it holds even if the input state $\rho$ is a tensor-product state diagonal in the computational basis.

An alternative interpretation of the clean and dirty computational model follows by inspecting Proposition~\ref{prop:conc-ladder-cost}
in the low error limit. Namely, considering for simplicity $n$ to be a power of $2$, then we have $(1-p)^{\Gamma_{n,n_d,L}}=(1-p)^{L_{\textrm{eff}}}$ and hence, if $p\ll 1$, then 
\begin{equation}\label{eq:rescaling}
    (1-p)^{L_{\textrm{eff}}}\sim (1-p\frac{n_d}{n}L)\sim(1-p\frac{n_d}{n})^{L}\,,
\end{equation}
up to first order in $p$.  Equation~\eqref{eq:rescaling} shows that adding clean qubits (left-hand side) is also approximately equivalent to a fully noisy system with an effective noise parameter rescaled by $n_d/n$. The identification  of the right-hand-side can be made precise from Proposition~\ref{prop:conc-ladder-cost} as $(1-p\frac{n_d}{n})^{L}$ is precisely the scaling coefficient one finds in an $L$-layered system where all qubits are subject to noise with a probability $p\frac{n_d}{n}$. As such, having fewer dirty qubits has a similar effect to reducing the effective noise in a system where all qubits undergo the same rate of depolarizing noise. This can be thought of as arising due to the fact that the CNOTs propagate the noise from dirty qubits through to the rest of the system. 

Thus, for sufficiently deep circuits, the localized contributions of error instances wash into the rest of the circuit, and the decoherence can be broadly characterized by the \textit{total number of error instances}, or equivalently effective error rate, and depth $L_{\textrm{eff}}$.  This interpretation is explored more in our numerical studies in Section \ref{sec:num}.
We note that our first setting can also be adapted to give similar results for Pauli observables. Specifically, by considering the reverse circuit with Pauli observables we observe the same exponential concentration. We discuss this further in Appendix \ref{sec:pauli-observables}.

\begin{figure}[t]
	\includegraphics[width=  \columnwidth]{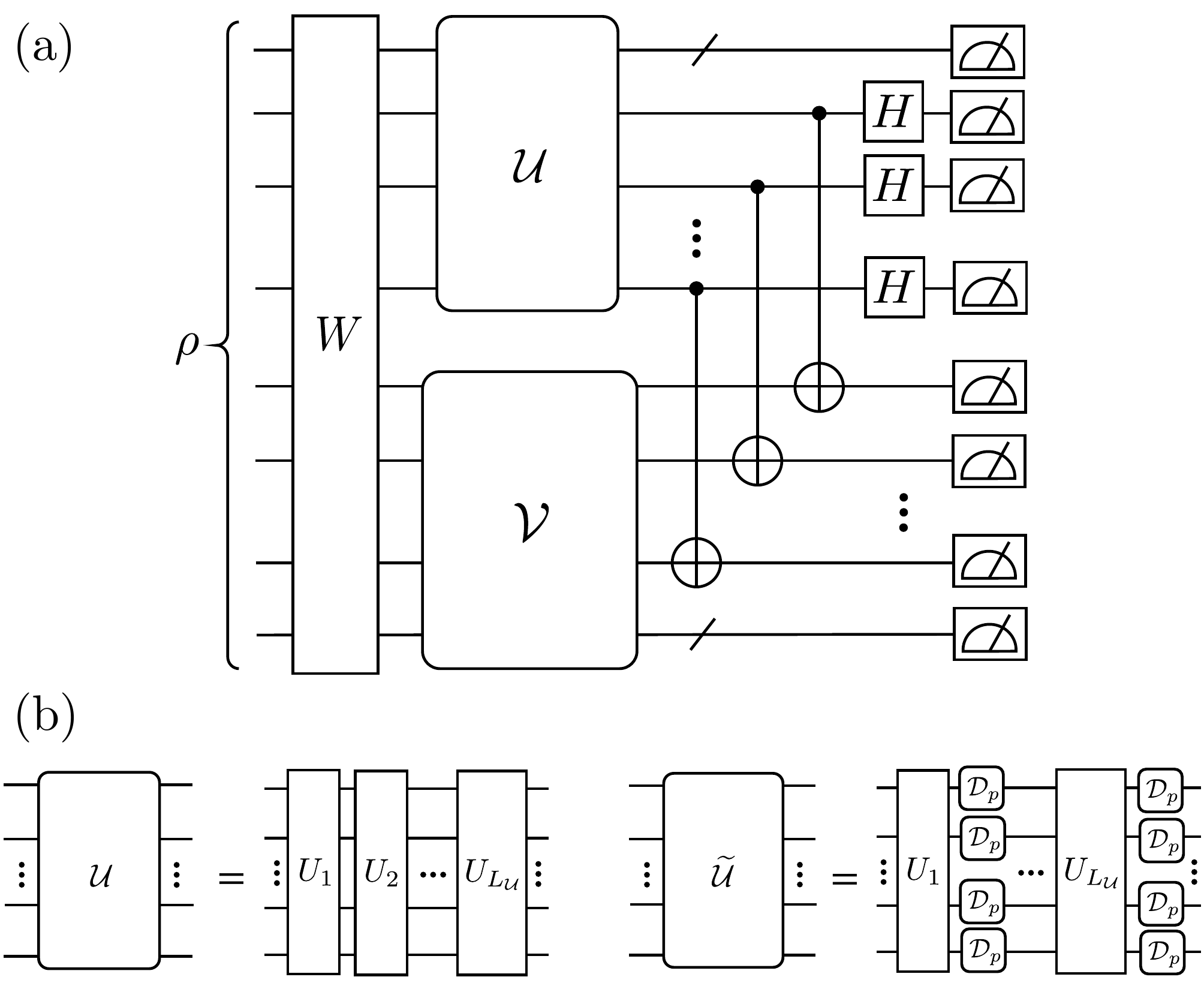}
	\caption{\textbf{Entangled registers} (a) In our second analytical example we show that entanglement can also cause local noise to spread globally. In this circuit, we consider an arbitrary input state $\rho$ that is acted on by some unitary $W$. We then partition the system into two subsystems, each of which undergoes its own unitary evolution, denoted as $\UC$ and $\VC$. In our analytics, we consider a modification where one of the subsystems undergoes noisy evolution whilst the other is noise-free. At the end of the circuit, a subset of the qubits undergoes an entangled measurement across the two subsystems.
	(b)  We define the depth of $\UC$ to be the number of non-parallelizable unitary layers $L_{\UC}$ in a given hardware implementation. When we consider a noisy modification $\UC \rightarrow \widetilde{\UC}$, we insert layers of local depolarizing noise with depolarizing probability $p$ in between each unitary layer. We define $L_{\VC}$ and $\widetilde{\VC}$ in an analogous manner.}\label{fig:entangleUV}
\end{figure}

In the second setting,  we show how local noise impacts the global system due to entangled measurements. This can also equally be considered to demonstrate the role of entanglement when considering the Heisenberg picture (i.e., considering the circuit acting in the  reverse order). We consider the circuit in Figure \ref{fig:entangleUV}(a) and suppose that the registers are split into two subsystems $A$ and $B$. Both subsystems are acted on by a global unitary $W$ (which can in principle even be separable across the cut $A|B$), before they undergo evolution according to some channel $\TC$ which is also separable across the cut $A|B$. One choice for $\TC$ is to consider noise-free unitary evolution which we denote as $\UC \otimes \VC$. In order to characterize the effects of noise, we can also modify these channels to account for hardware noise, as demonstrated in Figure \ref{fig:entangleUV}(b), in the sense that between each non-parallelizable layer of gates there is an instance of local depolarizing noise acting on each of the noisy qubits. We denote such a noisy modification as $\UC \rightarrow \widetilde{\UC}$ or $\VC \rightarrow \widetilde{\VC}$. Further, we denote the number of such unitary layers as $L_{\UC}$ and $L_{\VC}$ respectively. Finally, at the end of the circuit there is an entangled measurement across $m$ pairs of registers in the two subsystems, which we denote with the measurement operator $\BC_m(\dya{z_{2m}})$, where $\BC_m$ is the map corresponding to a Hadamard gate followed by a CNOT on $m$ pairs of qubits. We note that a special case of this circuit is the (global) Hilbert-Schmidt test circuit \cite{khatri2019quantum}. We find that the following proposition, proved in Appendix~\ref{app:th}, holds.

\begin{proposition}\label{prop:conc-UV-cost}
Consider the circuit in Figure \ref{fig:entangleUV}(a) partitioned into two subsystems $A$ and $B$. We have entangling operation $\BC_m$ on $m\geq1$ pairs of qubits across $A$ and $B$, noisy unitary evolution $\TC$ separable across the cut $A|B$, global unitary evolution $\WC$, and computational basis measurement $O_{z_n}=\dya{z_n}$; $z_n \in \{0,1\}^n$. For any input state $\rho$ and any computational basis measurement, the expectation value concentrates as
\begin{equation}
    \Tr\big[\BC_m^\dag\circ\TC\circ\WC(\rho)O_{z_n}\big] - \frac{1}{2^n} \leq  (1-p)^{L_\TC} \sqrt{P(\rho)}\,,
\end{equation}
where $P(\rho)$ is the purity of $\rho$ and we have defined
\begin{align}
    L_\TC = 
    \begin{cases}
    {L_{\UC}}\,,& \text{for}\ \TC = \widetilde{\UC} \otimes \VC\,, \\
    {L_{\VC}}\,,& \text{for}\ \TC = {\UC} \otimes \widetilde{\VC}\,, \\
    {L_{\UC}+L_{\VC}}\,,& \text{for}\ \TC = \widetilde{\UC} \otimes \widetilde{\VC} \,.
    \end{cases}     
\end{align}
\end{proposition}

Proposition \ref{prop:conc-UV-cost} shows the role entanglement (or entangled measurements) can play in spreading noise from a noisy subsystem to a clean subsystem. Namely, despite noise only occurring in one part of the circuit, the output distribution of the full system exponentially concentrates in the circuit depth. We also note that this occurs even if only one clean qubit is coupled to a dirty qubit ($m=1$). The proof follows from analysis in the Heisenberg picture and thus implies that entangling registers at the beginning of the computation also leads to an identical effect, which we demonstrate in Appendix \ref{sec:entangleUV}.

We note that, when only subsystem $A$ ($B$) is noisy, there is still exponential scaling, however, compared to the fully noisy setting the exponent is modified by factor $\frac{L_{\UC}}{L_{\UC}+ L_{\VC}}$ ($\frac{L_{\VC}}{L_{\UC}+ L_{\VC}}$). If we consider the unitaries to have linear depth in the number of qubits and scale proportionally with the same  factor (i.e. $L_{\UC}/n_A = L_{\VC}/n_B$), then the modification to the exponent is ${n_A}/{n}$ (${n_B}/{n}$). That is, the exponent is rescaled by the ratio of the number of dirty qubits to the total number of qubits, which is similar to the effect we find in Proposition \ref{prop:conc-ladder-cost}.

Propositions \ref{prop:conc-ladder-cost} and \ref{prop:conc-UV-cost} together show that despite some of the registers being clean, exponential concentration of expectation values in the circuit depth can still occur, in a similar spirit to that of NIBPs as presented in Ref.~\cite{wang2020noise}. Indeed, when we set $n_d=n$ we recover the same scaling.  Moreover, there is also exponential concentration with each increasing dirty qubit (for linear depth circuits, in the case of Proposition \ref{prop:conc-UV-cost}). We note that following the methods presented in Ref.~\cite{wang2020noise}, our results can also be generalized to classes of Pauli noise that have at least two types of Pauli error occurring with non-zero probability.

We remark that, despite the fact that our bounds are one-sided, there is some physical significance to the factor of $n_d/n$ that the exponent is rescaled to when transitioning from $n$ to $n_d$ dirty qubits. This is due to the fact that the exponent in our bounds corresponds to the Pauli string with the smallest amplitude decay. Thus, whilst we would not expect all circuits to exactly exhibit this scaling, we do expect it to be characteristic of circuits in the best case. Further details are given in the proofs of Propositions \ref{prop:conc-ladder-cost} and \ref{prop:conc-UV-cost}, which can be found in Appendices \ref{sec:ladder} and \ref{sec:entangleUV} respectively.

\subsection{Gradient scaling}\label{sec:gradient-scaling}

Our above results can also be readily transported to consider gradient scaling. First, we note that if either of the two above settings is modified to contain a trainable parameter (without breaking initial assumptions) and in settings amenable to the parameter shift rule, then Propositions \ref{prop:conc-ladder-cost} and \ref{prop:conc-UV-cost} automatically imply exponentially vanishing gradients in the depth of the circuit. Namely, this implies barren plateaus for linear depth circuits if $n_d$ scales linearly with $n$. This also implies barren plateaus for polynomial depth circuits of degree 2 or greater, for any $n_d>0$. Further, in more general settings where one cannot apply the parameter shift rule, we now explicitly demonstrate how gradient scaling results can be obtained for modified circuits, and present further details in Appendix \ref{sec:gradientscaling}. 

We consider a trainable unitary of the form
\begin{equation} \label{eq:trainable-unitary}
    Y(\thv)= W_0 \prod_{k=1}^K e^{-i \theta_{k} H_{k}} W_{k}\,,
\end{equation}
where $\{W_k\}_{k=0}^K$ are arbitrary fixed unitary operators and $\{H_k\}_{k=1}^K$ are Hermitian operators. We denote the channel that corresponds to this unitary as $\YC(\thv)$. In the following propositions, we modify the circuits in Section \ref{sec:expval-conc} by inserting $Y(\thv)$ in the middle of the circuit and constraining the input state to be a computational basis state.

\begin{proposition}\label{prop:ladder-grad}
Consider a cost function $C(\thv)=\Tr[\WC^{(n_d)}_{L_2}\circ\YC({\thv})\circ(\WC^{(n_d)}_{L_1})^{\dag}(\rho)O_{z_n}]$, where $\rho$ is a computational basis input state, $\WC^{(n_d)}_{L_i}$ denotes the channel corresponding to $L_i$ instances of CNOT ladders and local noise on $n_d$ qubits as presented in Figure \ref{fig:ladder}, and with computational basis measurement $O_{z_n}=\dya{z_n}$; $z_n \in \{0,1\}^n$. The partial derivative with respect to parameter $\theta_k$ is bounded as
\begin{equation}\label{eq:ladder-grad1}
    |\partial_{\theta_k} C| \leq (1-p)^{\Gamma_{n,n_d,L_1}+\Gamma_{n,n_d,L_2}}\|H_k\|_\infty\,,
\end{equation}
where we denote
\begin{align}
    \Gamma_{n,n_d,L_i} =
    \begin{cases}
    \left\lfloor L_i\cdot \frac{n_d }{2^{\floor{\log_2 n}}} \right\rfloor\,, & \text{for}\ n_d = 1 \,,\\
    \left\lfloor L_i\cdot \frac{n_d }{2^{\ceil{\log_2 n}}} \right\rfloor\,, & \text{for}\ n_d > 1\,,
    \end{cases} 
\end{align}
for $i \in \{1,2\}$.
\end{proposition}

\begin{proposition}\label{prop:entangleUV-grad}
Consider a cost function  $C(\thv)=\Tr[\BC^{\dag}_{m_2}\circ\TC_2\circ\YC(\thv)\circ\TC_1\circ\BC_{m_1}(\rho)O_{z_n}]$, where $\rho$ is a computational basis input state, $\BC_{m_i}$ is an entangling operation between ${m_i}$ pairs of qubits in subsystems $A_i$ and $B_i$ as considered in Proposition \ref{prop:conc-UV-cost}, $\TC_i$ denotes noisy unitary evolution separable across the cut $A_i|B_i$, and with computational basis measurement $O_{z_n}=\dya{z_n}$; $z_n \in \{0,1\}^n$. The partial derivative with respect to parameter $\theta_k$ is bounded as
\begin{equation}\label{eq:entangleUV-grad1}
    |\partial_{\theta_k} C| \leq (1-p)^{L_{\TC_1}+L_{\TC_2}}\|H_k\|_\infty\,,
\end{equation}
where we denote
\begin{align}
    L_{\TC_i} = 
    \begin{cases}
    {L_{\UC_i}}\,,& \text{for}\ \TC_i = \widetilde{\UC}_i \otimes \VC_i\,, \\
    {L_{\VC_i}}\,,& \text{for}\ \TC_i = {\UC_i} \otimes \widetilde{\VC}_i\,, \\
    {L_{\UC_i}+L_{\VC_i}}\,,& \text{for}\ \TC_i = \widetilde{\UC}_i \otimes \widetilde{\VC}_i \,.
    \end{cases}     
\end{align}
for $i \in \{1,2\}$, and $L_{\UC_i}$ and $L_{\VC_i}$ are defined in the same manner as $L_{\UC}$ and $L_{\VC}$ in Figure \ref{fig:entangleUV}.
\end{proposition}

Propositions \ref{prop:ladder-grad} and \ref{prop:entangleUV-grad} demonstrate that when modifying the settings considered in Section \ref{sec:expval-conc} to include a parameterized unitary, the partial derivative of measurement outcomes with respect to any parameter displays similar scaling to the concentration of expectation values. Specifically, if we assume that the largest singular values of the generators $\{H_k\}^K_{k=1}$ of the unitary $Y(\thv)$ scale at most polynomially in $n$, this implies NIBPs for linear depth circuits if $n_d$ also scales linearly in $n$, or for polynomial-depth circuits of degree 2 or greater for any $n_d>0$.

\section{Numerical Simulations}
\label{sec:num}
In this section we numerically investigate the behavior of the gradient of a cost function for the Hamiltonian Variational Ansatz (HVA) using both a single-qubit local depolarizing noise model and a model based on a trapped-ion quantum computer \cite{trout2018simulating, cincio2018learning}.  We start with an introduction of  the  HVA in Section~\ref{sec:HVA}, and follow with a presentation of the results in  Section~\ref{sec:num}.

\subsection{Hamiltonian Variational Ansatz}
\label{sec:HVA}

We consider an HVA  ~\cite{farhi2014quantum,hadfield2019quantum} for a task of variational  ground state  search for  a transverse field one-dimensional quantum Ising model.  The model  is given by  a Hamiltonian
\begin{equation}
    H = - \sum_{\langle i, j \rangle} X_i X_{j}  - g\sum_{i=1}^{n} Z_i\,.
\end{equation}
Here, $g$ is a constant, and $X,Z$ are Pauli matrices. We choose $g=1$ and assume Periodic Boundary Conditions (PBC).  
The ansatz  with $L$ layers takes the form 
\begin{equation}
\begin{split}
    \ket{\psi(\vec{\theta})} = \:  & e^{-i \theta_{2L} H_Z} e^{-i \theta_{2L-1} H_{XX}} \dots \\
 &   e^{-i \theta_2 H_Z} e^{-i \theta_1 H_{XX}}\ket{0}, \label{eq:HVA}
\end{split}
\end{equation}
where $H_{XX} = \sum_{\langle i, j \rangle} X_i X_j$, $H_Z = \sum_{i=1}^{n} Z_i$, and   $\vec{\theta}=(\theta_1, \theta_2, \dots,  \theta_{2L})$ are parameters. We decompose   $e^{-i \theta_k H_{XX}}$ and $e^{-i \theta_l H_{Z}}$ to entangling Molmer-S{\o}rens{\o}n gates and single qubit rotations, respectively,   as shown in Figure~\ref{fig:HVA}. 
We minimize the energy of the model using a cost function defined as
\begin{equation}\label{eq:cost-HVA}
C(\vec{\theta}) = \bramatket{\psi(\vec{\theta})}{H}{\psi(\vec{\theta})}, 
\end{equation}
and compute the gradient 
\begin{equation}
\nabla C(\vec{\theta}) = (\partial_{\theta_1} C(\vec{\theta}), \partial_{\theta_2} C(\vec{\theta}),  \dots,  \partial_{\theta_{2L}} C(\vec{\theta}) ), 
\label{eq:grad}
\end{equation} 
using a parameter shift rule~\cite{schuld2019evaluating}. The gradient (\ref{eq:grad}) is not affected by a barren plateau in the noiseless case~\cite{larocca2021diagnosing} making it a convenient setup for numerical investigation of the noise effects on the gradient scaling.    For  technical details on the gradient computation  see Appendix~\ref{app:shift}. 

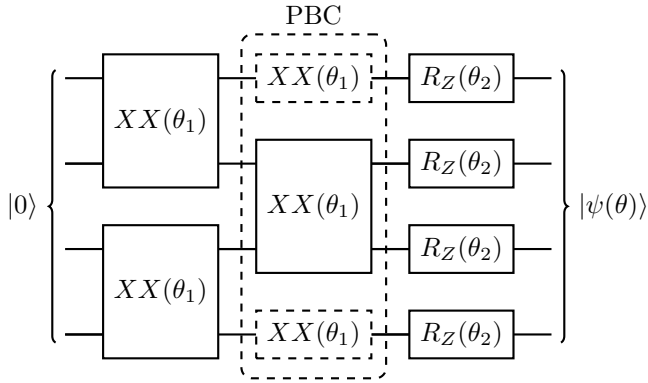
\begin{figure}[t]
\centering
\begin{quantikz}
    \lstick[wires=4]{$\ket{0}$}     & \gate[style={},wires=2]{XX(\theta_1)} & \gate[style={dashed}, wires=1]{XX(\theta_1)}\gategroup[4,steps=1,style={dashed,rounded corners, inner xsep=2pt},background]{{PBC}}& \gate[wires=1]{R_Z(\theta_2)} \qw  & \qw \rstick[wires=4]{$\ket{\psi(\mathbf{\theta})}$} \\
                                    & \qw                          & \gate[style={}, wires=2]{XX(\theta_1)}                                                                                                    & \gate[wires=1]{R_Z(\theta_2)} \qw    & \qw\\
                                    & \gate[style={},wires=2]{XX(\theta_1)} & \qw                                                                                                                             & \gate[wires=1]{R_Z(\theta_2)} \qw    & \qw\\
                                    & \qw                          & \gate[style={dashed}, wires=1]{XX(\theta_1)}                                                                                                    & \gate[wires=1]{R_Z(\theta_2)} \qw    & \qw 
\end{quantikz}
\caption{
    \textbf{The HVA ansatz.} A single layer HVA ansatz (\ref{eq:HVA}) for a 4-qubit system with periodic boundary conditions. We decompose the ansatz to Molmer-S{\o}rens{\o}n   $XX(\theta_1) = e^{- i\theta_1 X_j X_{j+1}}$ gates and single  qubit rotations   $R_Z(\theta_2) = e^{-i(\theta_2/2) Z_k}$ which are native gates of a trapped-ion quantum computer.        
    }\label{fig:HVA}
\end{figure}

\subsection{Numerical results}
\label{sec:num_res}

\begin{figure*}[ht]
    \centering
    \subfloat[Subfigure 1][Depolarizing noise model - gradient versus number of layers]{
    \includegraphics[width=0.48\textwidth]{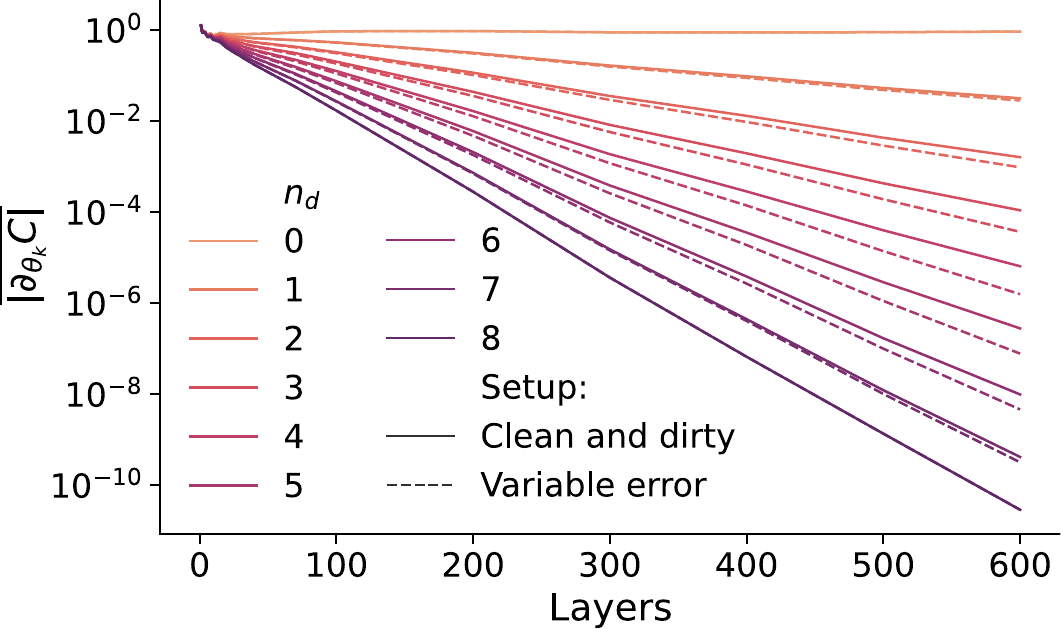}
    \label{fig:DepolLayer}}\vspace{0mm}
    \subfloat[Subfigure 2][Realistic noise model - gradient versus number of layers]{
    \includegraphics[width=0.48\textwidth]{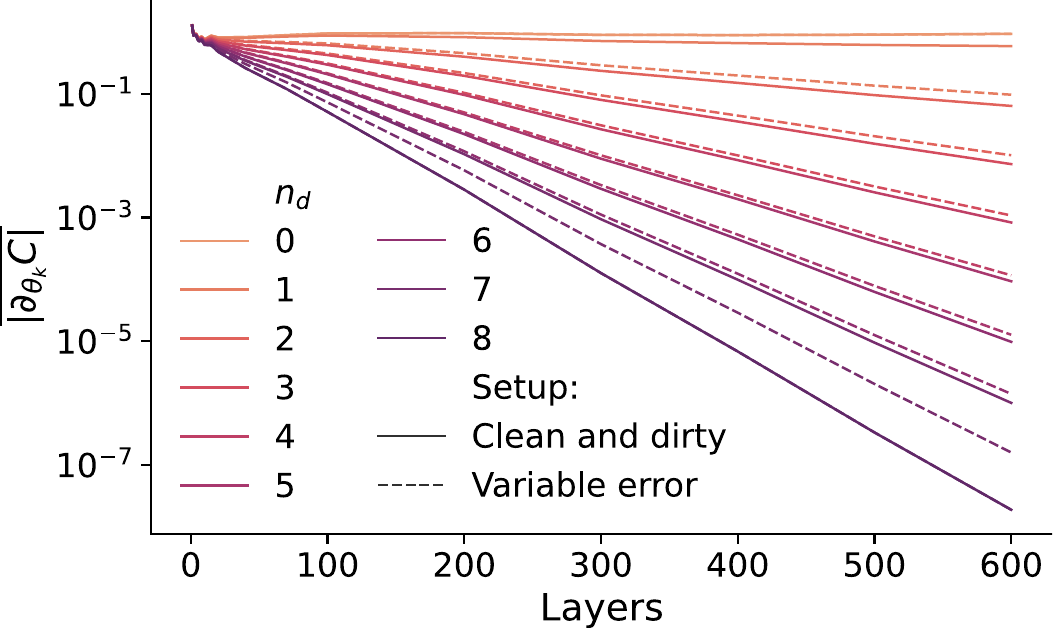}\label{fig:RealLayer}}
    \caption{
    \textbf{HVA gradient scaling versus the number of HVA layers.}  We plot the mean absolute value of the partial derivative of the cost function with respect to a parameter $\theta_k$  $\overline{|\partial_{\theta_k} C|}$   for $L=1,\ldots,600$ and $n=8$. The average is taken over both  $28$ random sets of parameters $\vec{\theta}$ and the parameters $\theta_k$. In (a) we show results for  local  depolarizing noise. In (b)  we show results for the trapped-ion  realistic noise model. The solid lines are obtained for $n_d=0,1,\dots,8$.  In the plots larger $n_d$ is plotted with brighter colors.  The dashed lines are obtained for  the case of  $n$ dirty qubits and   error rates of the noisy gates  scaled by a factor $f=1/8, 2/8, \dots, 7/8$ with respect to their noise rates for the aforementioned simulations.  We plot the results for $n_d$ and $f=n_d/n$ with the same color.    
     }   \label{fig:numericalResultLayer}
\end{figure*}

We examine the scaling of the averaged absolute value of the partial derivative of the cost function with respect to a parameter $\theta_k$ $\overline{|\partial_{\theta_k} C|}$  for the cost function of Eq.~\eqref{eq:cost-HVA} in a system with $n=8$ qubits and for an HVA ansatz with the number of layers $L$ increasing from $1$ to $600$.
The chosen range of $L$ allows to showcase the asymptotic scaling of  $\overline{|\partial_{\theta_k} C|}$ with increasing $L$.
The average is uniformly  taken over both random sets of parameters $\vec{\theta}$ and the parameters $\theta_k$.   We consider   $n_d= 0, 1, 2, \dots, n$ in both  local  depolarizing noise and a realistic trapped-ion noise model~\cite{trout2018simulating, cincio2020machine}. 

The depolarizing noise model applies a local depolarizing channel after each gate if the qubit is designated to be dirty. In the case of the $XX$ gate, noise is represented by applying a single qubit depolarizing channel to each dirty qubit. The realistic model is based on real-world machines~\cite{trout2018simulating} and consists of single-qubit noise channels for single-qubit gates acting on the dirty qubits, and a correlated two-qubit noise channel after $XX$ gates. The $XX$ noise is applied only when both qubits are dirty. A detailed description of the noise models, including figures and parameters, can be found in Appendix~\ref{app:noise}.

We assume that the dirty qubits form a contiguous block as in Figure~\ref{fig:NCsetup}. We gather  results in Figure~\ref{fig:numericalResultLayer}. Here we can see that  $\overline{|\partial_{\theta_k} C|}$  decays exponentially with $L$ for $n_d>0$ (as evidenced by a straight line in the log-linear plot). Furthermore, we can see that the averaged gradient also decreases exponentially with increasing $n_d$. This can be seen from the fact that for a fixed $L$, increasing $n_d$  decreases the value of  $\overline{|\partial_{\theta_k} C|}$  by a constant factor in the log-linear plot. Hence, we find in both cases that for $n_d>0$,  $\overline{|\partial_{\theta_k} C|}$  decays exponentially with $L$ and $n_d$. Namely, we heuristically observe that $ \overline{|\partial_{\theta_k} C|}\in\OC(a^{Ln_d})$ for some $a<1$.  These results reflect the exponential scaling in $L$ and $n_d$ obtained in our theoretical results from Propositions~\ref{prop:conc-ladder-cost} and \ref{prop:conc-UV-cost}, thus indicating that the scaling in our propositions  might be present in realistic models.

\begin{figure*}[ht]
    \subfloat[Subfigure 1][Depolarizing noise model - gradient versus $L_{\textrm{eff}}$]{
    \includegraphics[width=0.48\textwidth]{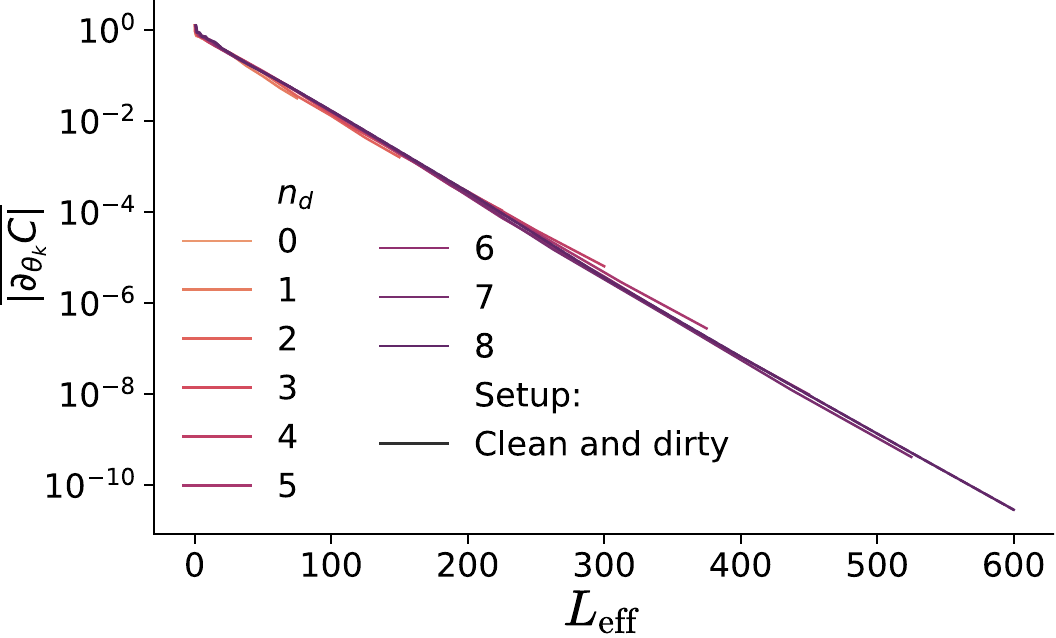}
    \label{fig:DepolScaled}}\vspace{0mm}
    \subfloat[Subfigure 2][Realistic noise model - gradient versus total error rate]{
    \includegraphics[width=0.48\textwidth]{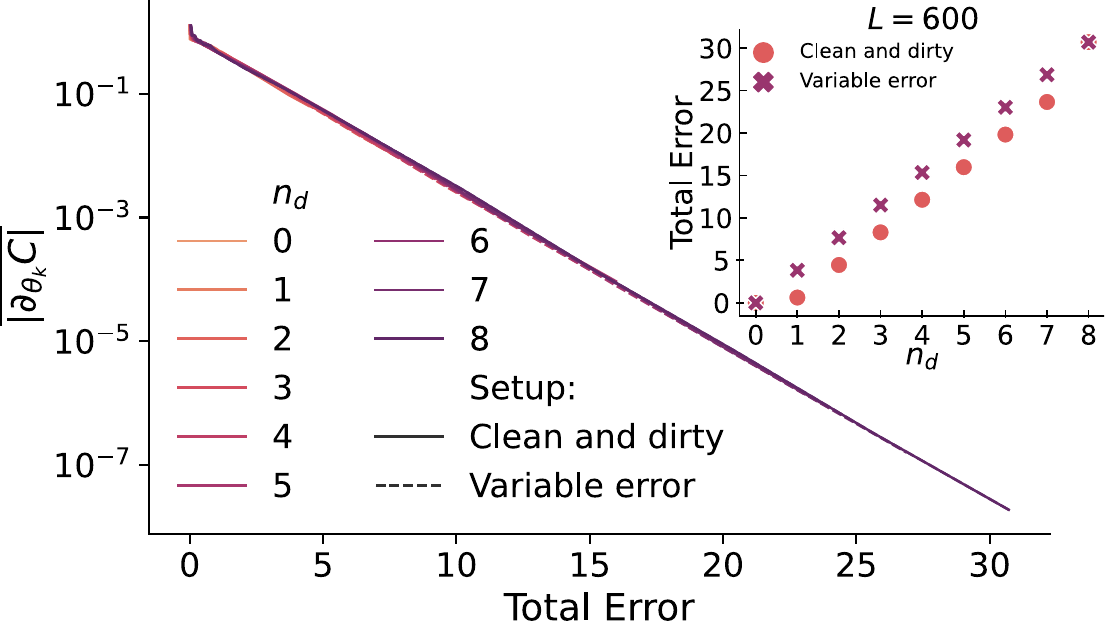} \label{fig:RealScaled}}
    
    \caption{
    \textbf{HVA gradient scaling versus the total error rate.} The plots show the mean absolute value of the partial derivative of the cost function $\overline{|\partial_{\theta_k} C|}$ from Figure~\ref{fig:numericalResultLayer} plotted versus the total error, which in the depolarizing case is proportional to $L_{\textrm{eff}}=L\frac{n_d}{n}$. The solid lines show  the clean and dirty setup, while the dashed ones show  the variable error setup, with darker lines representing lower noise levels.  Both the depolarizing noise model (a) and the realistic noise  (b) results show an exponential decay   of the   mean absolute value of the  partial derivative with the total error rate    up to  small deviations. The inset in (b) shows the error rate for the two setups and  the realistic noise model in relation to $n_d$  demonstrating that it is approximately linear. Consequently, in both cases, we obtain that the mean gradient decays approximately exponentially with $n_d$.  
    }
    \label{fig:numericalResultScaled}
\end{figure*}

Here we can also verify to what extent Eq.~\eqref{eq:rescaling} holds in the considered realistic scenario. Namely, we perform numerical simulations for the same HVA ansatz as previously analyzed but now all the qubits are dirty and the local depolarizing probability is rescaled according to Eq.~\eqref{eq:rescaling}. For instance, in the case of  depolarizing noise, we consider the error rates
\begin{equation}
 p \rightarrow p f, \quad\text{with } f=1/n, 2/n, \dots, (n-1)/n, 1.
\end{equation}
The error rates in the case of the trapped-ion noise model are scaled analogously, see Appendix~\ref{app:noise} for details. 
 
Figure~\ref{fig:numericalResultLayer} shows that the results for $n_d<n$ (solid lines), with the local depolarizing and the trapped-ion noise  models, are similar to the case when all qubits are dirty and the noise is rescaled by $f=n_d/n$ and  $f=(n_d-1)/n$, respectively (dashed lines). The difference in $f$ comes from the fact that in the case of depolarizing noise, the number of noise channels per layer of the ansatz is proportional to $ n_d$,
while in the case of the trapped-ion noise model we have  $n_d-1$   noisy Molmer-S{\o}rens{\o}n gates per layer, as explained in Appendix~\ref{app:noise}. 
Therefore, our numerics demonstrate that adding clean qubits is approximately equal to rescaling the error rate by a factor $n_d/n$ as in Eq.~\eqref{eq:rescaling}. In Appendix~\ref{app:additionalResults} we additionally analyze the gradient behavior for $n=4,6$ qubits, obtaining the same conclusions as here for $n=8$.

We further explore the behavior of the averaged derivative of the cost function with respect to the parameters in Figure~\ref{fig:numericalResultScaled} by plotting   $\overline{|\partial_{\theta_k} C|}$ versus the total circuit error rate, which we define as the sum of the error rates of all gates in a circuit. For both the local depolarizing and the trapped-ion noise models we obtain with good accuracy a collapse of $\overline{|\partial_{\theta_k} C|}$ plotted versus the  circuit error rate for all considered values of $n_d$ and    $f$. In Appendix~\ref{app:additionalResults} we conduct a similar analysis  for $n=4,6$ obtaining the same conclusion. We see small deviations in the depolarizing case which appear to increase with increasing system size. Furthermore, in the same appendix, we also explore the clean and dirty setup for low circuit depths ($L=1-30$) up to $n=10$  in the realistic noise model, obtaining similar behavior as for the higher depths.

We note that the deviations  are within one standard deviation of the mean gradient which may indicate that there are finite sample artifacts. Nevertheless, to elucidate their origin a more systematic investigation would be necessary, which  we leave to future work.  We find that  for  the realistic noise model deviations are much smaller resulting in nearly perfect collapse. Finally, we note that in the case of the depolarizing noise, the total error is proportional to $n_d/n$, while in the case of the realistic noise, it scales approximately linearly   with $n_d/n$  as shown in an inset of  Figure~\ref{fig:RealScaled}. Therefore, the obtained collapses provide further examples of exponential suppression of the gradient by  a ratio of $n_d/n$.

\section{Discussion}

We have found that partitioning a quantum computer into noiseless ``clean'' qubits and noisy ``dirty'' qubits typically has an analogous effect to lowering the overall error rate by the ratio $\frac{n_d}{n}$. Consequently, our model of clean qubits cannot avoid exponential scaling effects due to noise such as barren plateaus, and even a single dirty qubit will ruin the bunch eventually. As it has previously been shown that standard error mitigation techniques also cannot avoid exponential scaling effects due to noise~\cite{wang2021can,takagi2021fundamental}, our work provides further evidence that a fully noiseless system may be ultimately required to avoid such scaling. We note that even though we exclusively look at gradient scaling in our numerics, this also implies the concentration of expectation values due to the results of~\cite{arrasmith2021equivalence}. This means that our results have implications for the scalability of near-term algorithms in general, not just in the realm of training parameterized circuits.

On the other hand,  we found numerically that each additional clean qubit exponentially suppresses the rate of gradient decay with respect to depth.  Therefore, the inclusion of clean qubits helps to mitigate NIBPs and exponential concentration of expectation values due to noise.  Our analytical results point to this being a more general phenomenon in partially noisy circuits.  We note that here, in order to  numerically investigate the  scaling of trainability at large depths, we assumed that the error rates of the noisy qubits are larger than the ones expected from future machines. Their reduction  should enable even more impressive increases in computational reach  than the ones shown in our numerics.

We expect this work to lay the groundwork for potential explorations into the post-NISQ era, which is seemingly drawing closer with ever-larger qubit counts and decreasing error rates. Our results indicate that  quantum computation with both error-corrected and noisy qubits can mitigate the exponential concentration issues. Therefore, it motivates the search for practically viable realizations  of this computational paradigm.   The clean and dirty setup is an idealized model of such a quantum computer, as a clean qubit implementation requires the correction of every error occurring at a logical qubit.

A natural follow-up study would be an extension of our model that takes into account that real-world logical qubits are characterized by logical error rates. Such a model would also enable a more comprehensive study of strategies involving combining logical qubits with different logical error rates and their comparison to using the same code distance for all the device qubits. In order to better model real-world logical qubits, it is essential to consider the gate-dependent errors, as non-Clifford gates pose significantly greater implementation challenges than Clifford gates. 

Another line of follow-up research would be to consider the application of the clean and dirty setup to algorithms that involve some qubits subject to many more gates than the others, such as the quantum phase estimation algorithm~\cite{somma2020quantum}. For such qubits, one naturally obtains higher effective error rates, and potentially better error suppression when just a fraction of qubits are corrected. Therefore, in such cases, the clean and dirty setup may result in even better error suppression than in the cases investigated here. Another worth-exploring scenario that may result in the clean and dirty setup as an effective model is a coupling of a sensor built from noisy qubits~\cite{danilin2022quantum} to logical qubits processing information obtained from the sensor~\cite{lauk2020perspectives}.  Finally, one can also ask about using the clean and dirty computational paradigm to design custom quantum algorithms that can maximize quantum advantage in the post-NISQ era.

\section{Acknowledgements}

We thank Andrew Arrasmith, Burak Sahinoglu, Chenfeng Cao, and Hsin-Yuan Huang  for their helpful discussions. We also thank Tom O'Leary for producing an open-source Qiskit version of the clean and dirty setup for the repository. This work was supported by the Quantum Science Center (QSC), a National Quantum Information Science Research Center of the U.S. Department of Energy (DOE). DB was supported by the U.S. DOE through a quantum computing program sponsored by the Los Alamos National Laboratory (LANL) Information Science \& Technology Institute and by the European Union’s Horizon 2020 research and innovation programme under the Marie Sk\l{}odowska-Curie grant agreement No. 955479. SW was supported by the UKRI EPSRC grant No.~EP/T001062/1 the Samsung GRP grant. MHG and LC were initially supported by the U.S. DOE, Office of Science, Office of Advanced Scientific Computing Research, under the Quantum Computing Application Teams~program. Piotr C. acknowledges initial support from the Laboratory Directed Research and Development (LDRD) program of Los Alamos National Laboratory (LANL) under project number 20190659PRD4 with subsequent support by  the National Science Centre (NCN), Poland under project 2019/35/B/ST3/01028. MC was supported by the LDRD program of LANL under project number 20210116DR. PJC also acknowledges initial support from the LANL ASC Beyond Moore's Law Project.

\section{Code availability}
The scripts and data used to create all the plots and to recreate some of the data can be found in Appendix~\ref{app:repository}.

\bibliographystyle{quantum}
\bibliography{main}

\clearpage
\newpage

\onecolumngrid
\appendix
\setcounter{section}{0}
\makeatletter
\@addtoreset{proposition}{section}
\makeatother

\section{Theoretical results}
\label{app:th}

\subsection{Preliminaries}

\begin{definition}[Depolarizing noise]
The action of a single-qubit depolarizing noise channel $\DC_p$ with error probability $p$ on a single-qubit quantum state $\rho$ can be written as
\begin{align}
    \mathcal{D}_p(\rho)     &= (1-p)\rho + p\frac{\id}{2}\,,\label{eq:depol}
\end{align}
where $\id$ is the $2\times 2$ identity matrix.
\end{definition}

\begin{definition}[Pauli strings]\label{def:Pstrings}
We define the set of non-identity Pauli tensor product strings $\mathrm{P}_n$ as
\begin{equation}
    \mathrm{P}_n = \{\id, X, Y, Z\}^{\otimes n}/\{\id^{\otimes n}\}\,.
\end{equation}
\end{definition}

\emph{Action of local depolarizing noise on Pauli strings.} We note that for any $\sigma \in \mathrm{P}_{n}$ we have
\begin{align}\label{eq:noise-on-Pauli}
    \DC_p^{\otimes n}(\sigma) = q_{\sigma} \sigma\,, 
\end{align}
where $|q_{\sigma}|\leq 1-p$.

\subsection{Useful Lemmas}

\begin{lemma}[Connecting purity to Hilbert-Schmidt norm]\label{lem:purity}
For any given state $\rho$, the purity $P(\rho)=\Tr[\rho^2]$ satisfies 
\begin{equation}
    P(\rho) - \frac{1}{2^n} = \Big\|\rho-\frac{\id}{2^n}\Big\|^2_2\,,
\end{equation}
where $\|.\|_2$ denotes the Schatten 2-norm.
\end{lemma}
\begin{proof}
We have 
\begin{align}
    \Tr[\rho^2]
    &= \Tr\left[ (\rho - \frac{\id}{2^n})^2 + \frac{2}{2^n}\rho - \frac{1}{2^n}\frac{\id}{2^n} \right] \\
    &= \Tr\left[ (\rho - \frac{\id}{2^n})^2\right] + \frac{1}{2^n} \\
    &= \Big\|\rho-\frac{\id}{2^n}\Big\|^2_2 + \frac{1}{2^n}\,,
\end{align}
where in the first line we have added and subtracted $ \frac{2}{2^n}\rho - \frac{1}{2^n}\frac{\id}{2^n}$, in the second line we have explicitly computed the trace for the last two terms, and in the final line we have used the definition of the Schatten 2-norm.
\end{proof}

\begin{lemma}(2-norm of commutators.) \label{lem:commutator}
   Let $X$ and $Y$ be complex matrices. Let $\|\cdot\|_2$ and $\|\cdot\|_{\infty}$ denote the Schatten 2- and $\infty$-norms respectively. We have
    \begin{equation}
        \big\|[X,Y]\big\|_2 \leq 2 \,\big\|X\big\|_2 \big\|Y\big\|_\infty\,.
    \end{equation}  
\end{lemma}
\begin{proof}
Writing out the commutator explicitly, we have 
\begin{align}
    \big\|XY - YX\big\|_2 &\leq \big\|XY \big\|_2 + \big\|YX \big\|_2\\
    &\leq \big\|X \big\|_2 \big\|Y \big\|_{\infty} + \big\|Y \big\|_{\infty} \big\|X \big\|_2\,,
\end{align}
where the first inequality is due to the triangle inequality, and the second inequality is an application of the tracial matrix H\"older inequality \cite{baumgartner2011inequality}.
\end{proof}

\subsection{Model 1 - CNOT ladders} \label{sec:ladder}

In this section we consider a circuit consisting of CNOT ladders as presented in Figure \ref{fig:ladder_heisenberg}(a). After each ladder, a fraction $n_d/n$ of the qubits experiences single-qubit depolarizing noise $\DC_p$. 

We first provide a brief roadmap of our proof methods. In order to analyze such circuits, we consider the Heisenberg picture and bound the output purity of the reverse circuit (Figure \ref{fig:ladder_heisenberg}(b)). As we only consider computational basis measurements, we only need to consider Pauli tensor product strings of the form $\{\id, Z\}^{\otimes n}$. We will adopt the notation $A\,...\,B$ to represent a tensor product string $A \otimes ... \otimes B$. Due to Eq.~\eqref{eq:noise-on-Pauli}, local depolarizing noise on $n_d$ qubits will decrease the amplitude of certain strings in $\{\id, Z\}^{\otimes n}$. Specifically, wherever a string has a "$Z$" on a noisy register, that string's amplitude will decay by a factor $(1-p)$ for each instance of local depolarizing noise $\DC_p$. The goal of our analysis will be to observe patterns in Pauli strings as they are acted on by CNOT ladders and to keep track of the "best case" Pauli string that is least affected by noise. This corresponds to the Pauli string that has the fewest number of "$Z$"s appearing on noisy registers. This can then be used to establish upper bounds on expectation value concentration.

\begin{definition}[CNOT ladder]
We denote the mapping corresponding to one CNOT ladder (consisting of CNOT gates between all pairs of adjacent qubits in a 1D chain of $n$ qubits as in Figure \ref{fig:ladder_heisenberg}(b)) as $\textrm{\emph{CNOT}}_n$.
\end{definition}

\begin{definition}[CNOT ladder + noise]\label{def:cnot-noise}
We denote the channel corresponding to a CNOT ladder plus depolarizing noise on the first $n_d$ qubits as $\WC^{(n_d)\dag} :=  \textrm{\emph{CNOT}}_n\circ(\DC_p^{\otimes n_d} \otimes \IC^{n-n_d})$. Further, we denote the channel corresponding to $L$ instances of this channel (i.e. the full circuit in Figure \ref{fig:ladder_heisenberg}(b)) as $\WC^{(n_d)\dag}_L := (\WC^{(n_d)\dag})^L$.
\end{definition}

\begin{figure}[h]
    \begin{center}
	\includegraphics[width= .8 \columnwidth]{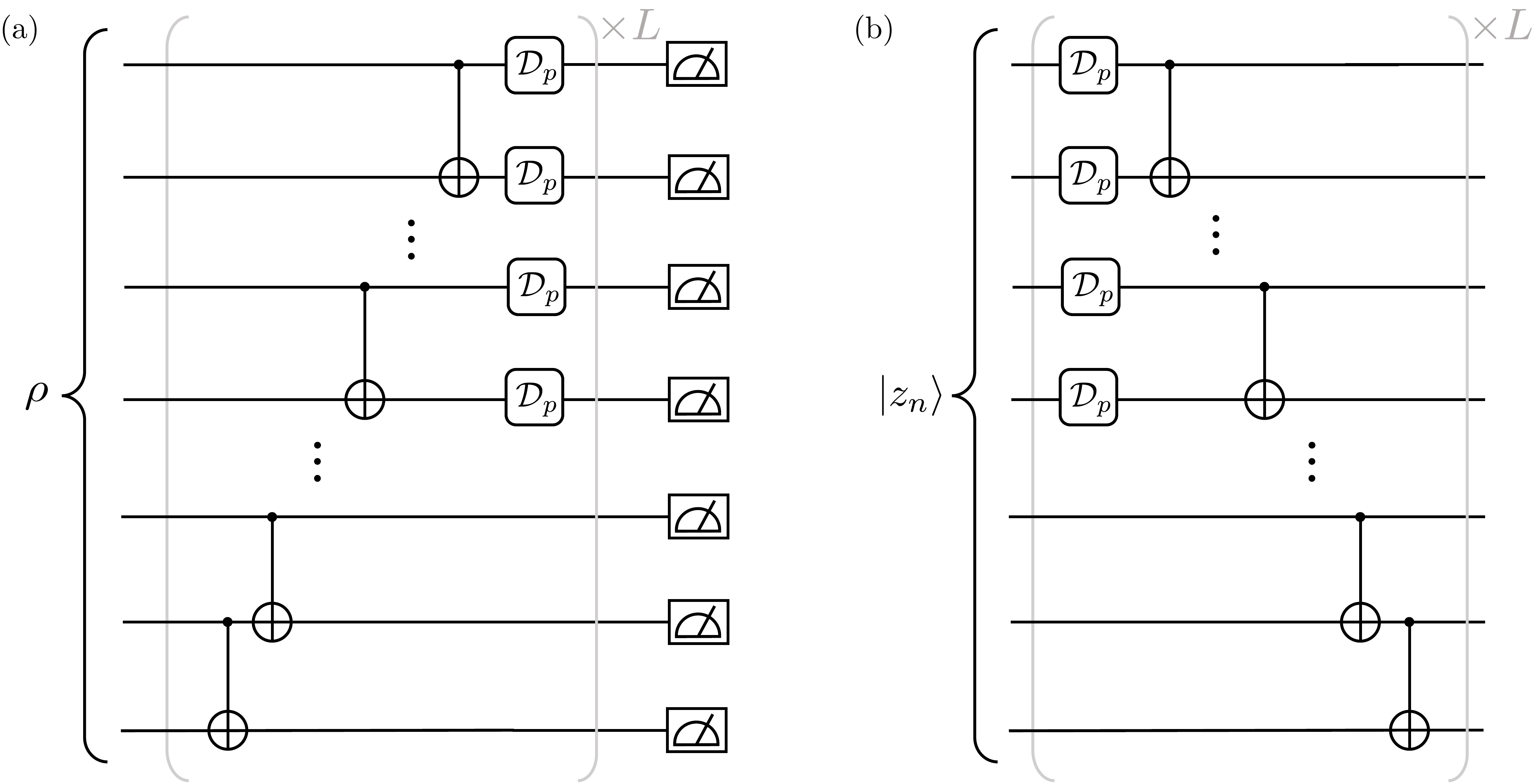}
    \end{center}
	\caption{\textbf{} (a) Our goal is to consider a circuit with a repeating ansatz of $L$ layers of a CNOT ladder. After each layer, the qubits on the first $n_d$ registers go through depolarizing noise channel $\DC_p$. At the end of the circuit, we measure in the computational basis. We denote $\rho$ as the input state. (b) In order to analyze such circuits, we consider the Heisenberg picture and bound the output purity of the reverse circuit, with computational basis input states.
	}\label{fig:ladder_heisenberg}
\end{figure}

As a preliminary example, we consider the mapping of the string $Z\otimes\id\otimes Z\otimes Z$ under such CNOT ladders:

\begin{align}\label{eq:examplemapping}
    Z \id ZZ \xrightarrow{\textrm{CNOT}_4} ZZ\id Z \xrightarrow{\textrm{CNOT}_4} \id ZZZ \xrightarrow{\textrm{CNOT}_4} Z\id \id Z \xrightarrow{\textrm{CNOT}_4} Z \id ZZ \,.
\end{align}

We note that after 4 iterations, the $Z\id ZZ$ is mapped back to itself. We will refer to this as a string with period 4.

\begin{definition}[Periodicity of a string]
We define the periodicity of a string $X$ under a map $\MC$ as the smallest value $p$ such that $X\xrightarrow{\MC^p}X$, where $\MC^p$ denotes  $p$ applications of the map $\MC$.
\end{definition}

We make two observations about the chain of mappings in \eqref{eq:examplemapping}. First, the string is mapped only to other strings in the set $\{\id, Z\}^{\otimes n}$. Second, the mapping is cyclic. Both observations can be noted to be a direct consequence of the fact that all strings from $\{\id, Z\}^{\otimes n}$ are mapped bijectively to another element of the set under $\textrm{CNOT}_n$, which we will now formally show.

\begin{lemma}
The set $\{\id, Z\}^{\otimes n}$ is mapped to itself under $\emph{\textrm{CNOT}}_n$ bijectively.
\end{lemma}

\begin{proof}
We can check the mappings for $n=2$ and a single CNOT gate:
\begin{align}
    &\id \id \xrightarrow{\textrm{CNOT}} \id \id \,,\\
    &\id Z \xrightarrow{\textrm{CNOT}} Z Z \left( \xrightarrow{\textrm{CNOT}} \id Z\right)\,,\\
    &Z Z \xrightarrow{\textrm{CNOT}} \id Z\,,\\
    &Z \id \xrightarrow{\textrm{CNOT}} Z \id\,.
\end{align}
As $\textrm{CNOT}_n$ is composed of $n-1$ CNOT maps, it follows that strings in $\{\id, Z\}^{\otimes n}$ are mapped to strings in $\{\id, Z\}^{\otimes n}$ by $\textrm{CNOT}_n$. Finally, the map is unitary and so has an inverse, and thus the map is bijective.
\end{proof}

Returning to our former example, we observe the elements of the cyclic mapping of $Z\id ZZ$ can be written as a slanted list

\begin{center}
\Longstack{
Z\x \id\x Z\x Z\x \hphantom{\id} \hphantom{\id}\x \x \hphantom{\id}\x \hphantom{\id}\\
Z\x Z\x \id\x Z\x \hphantom{\id}\x  \hphantom{\id}\x \hphantom{\id}\\
\id\x Z\x Z\x Z\x \hphantom{\id}\x  \hphantom{\id}\\
\underline{\vphantom{,} Z\x \id\x \id\x Z\x} \hphantom{\id} \\
Z\x \id\x Z\x Z}
\vline\quad
\Longstack[l]{
r=0\\
r=1\\
r=2\\
r=3\\
r=4
}
\end{center}
where we denote the row of the list with $r$, and the line marks the end of a cycle. We now make a key remark.

\begin{definition}[Inverted Binary Pascal Triangle]
We define the inverted binary Pascal triangle as the triangle of rows of binary strings, where the $L$ elements in the $k$th row are uniquely generated by $L+1$ elements of the $(k-1)$th row, with each element constructed by summing the element directly above and to the left, with the element directly above and to the right. An inverted binary Pascal triangle with a length $n$ string in the $0$th row has in total $n$ rows, ending in a row with a length $1$ string.
\end{definition}

\begin{remark}
The list of cyclic mappings of any element $L \in \{\id, Z\}^{\otimes n}$ under $\emph{\textrm{CNOT}}_n$ has a one-to-one correspondence with the inverted binary Pascal triangle whose string in row $0$ is generated by taking $L \otimes \id^{\otimes n-1}$ and applying the map $Z\leftrightarrow1, \id\leftrightarrow0$.
\end{remark}

As an example for the above remark, we can observe this correspondence for the string $Z\id ZZ$ (we truncate the triangle at $r=3$).

\begin{center}
\Longstack{
Z\x \id\x Z\x Z\x \id\x \id\x \id\\
Z\x Z\x \id\x Z\x \id\x \id\\
\id\x Z\x Z\x Z\x \id\\
Z\x \id\x \id\x Z\\
}
\hspace{10pt} \Longstack{\rightarrow \\ \\}\hspace{10pt}
\Longstack{
1\x 0\x 1\x 1\x 0\x 0\x 0\\
1\x 1\x 0\x 1\x 0\x 0\\
0\x 1\x 1\x 1\x 0\\
1\x 0\x 0\x 1\\
}
\end{center}

We now prove some general statements about inverted binary Pascal triangles.

\begin{lemma}\label{lem:binary-addition}
Given an inverted binary Pascal triangle with length $n=2^x+1;\; x \in \mathbb{Z}$ of the form $A\,...\,B$, the single entry in row $n-1$ has value $A \oplus B$. That is, all triangles of this size take the form
\begin{center}
\Longstack{
A\x ... \x B\\
...\x ...\\
...\\
A \oplus B
}
\hspace{20pt}
\vline\quad
\Longstack[l]{
r=0\\
\;\,\ ...\\
\;\,\ ...\\
r=n-1
}
\end{center}
\end{lemma}

\begin{proof}
We prove the lemma by induction. The primal case of $n=2$ can be checked by explicitly writing the inverted Pascal triangles as
\begin{center}
\Longstack{
0\x 0\\
0
}\ ,\quad\quad
\Longstack{
0\x 1\\
1
}\ ,\quad\quad
\Longstack{
1\x 0\\
1
}\ ,\quad\quad
\Longstack{
1\x 1\\
0
}\ .
\end{center}
Now suppose that strings of some longer length $k$ also always satisfy this property. For instance, for length $k$ string $A\ ...\ B$, we have
\begin{center}
\Longstack{
A\x ... \x B\\
...\x ...\\
...\\
A \oplus B
}
\hspace{20pt}
\vline\quad
\Longstack[l]{
r=0\\
\;\,\ ...\\
\;\,\ ...\\
r=k-1
}
\end{center}
Further, we consider another length $k$ string $B\ ...\ C$ that satisfies this property. This implies that the merged length $2k-1$ string $A\ ...\ B\ ...\ C$ has an inverted Pascal triangle of the form 
\begin{center}
\Longstack{
A\y\z ... \y\z B \z\y ... \y\z C \\
...\y\z ...\y\z ... \y\z ...\\
...\z\y \hphantom{B} \y\z ...\\
A \oplus B\ ...\ B \oplus C\\
...\x ...\\
...\\
A \oplus C
}
\hspace{20pt}
\vline\quad
\Longstack[l]{
r=0\\
\;\,\ ...\\
\;\,\ ...\\
r=k-1\\
\;\,\ ...\\
\;\,\ ...\\
r=2k-2
}
\end{center}
By merging strings together inductively, we see that strings of a particular length satisfy this property. The lengths satisfy a sequence with values
\begin{equation}
    a(m) = 2a(m-1) -1\quad ;\ m\in \mathbb{Z}\,,\ a(0)=2\,.
\end{equation}
It can be verified that this sequence has values $a(m) = 2^m +1$, which completes the proof.
\end{proof}

\begin{corollary}[Bit-wise addition of strings]\label{cor:bitwise}
Given a string composed of adjacently joining two strings of length $n \in 2^x\ ;\ x \in \mathbb{Z}$, starting in row $0$, the first $n$ entries of the string in row $n$ is the binary bit-wise addition of both strings.
\end{corollary}

\begin{proof}
Using Lemma \ref{lem:binary-addition}, this can be seen by considering two strings $A\ ...\ B$, $C\ ...\ D$, each of length $n$, and by drawing out the triangle corresponding to the string $A\ ...\ B\ C\ ...\ D$:
\begin{center}
\Longstack{
A\y\z ... \y\z BC \z\y ... \y\z D \\
...\y\z ...\,\ \hphantom{BC}\,\ ... \y\z ...\\
...\z\y \hphantom{BC} \y\z ...\\
A \oplus C\z ...\z B \oplus D
}
\hspace{20pt}
\vline\quad
\Longstack[l]{
r=0\\
\;\,\ ...\\
\;\,\ ...\\
r=n
}
\end{center}
\end{proof}

\begin{lemma}[Maximum cycle period]\label{prop:maxcycle}
Under the CNOT ladder $\mathrm{CNOT}_n$, any length $n$ string has a cycle of periodicity $p = 2^x ;\ x \in \mathbb{Z}$, where
\begin{equation}
    x \leq \ceil{\log_2 n}\,.
\end{equation}
That is, the periodicity is an integer power of 2, and at most the closest power of 2 to $n$ that is greater than $n$.
\end{lemma}

\begin{proof}
Suppose we start with a length $n$ string $A\ ...\ B$. Now we can append arbitrary elements to the start of the string such that the total length is $2^{\ceil{\log_2 n}}$, and we denote this new string as $A'\; ...\; A\; ...\; B$. Due to Corollary \ref{cor:bitwise}, we have the following triangle

\begin{center}
\Longstack{
A'\; ...\; A\; ...\; B\x 0\x ...\x 0\ \ \\
...\x \x\z ...\x ...\x ...\\
...\x \x \:\,\ ...\x\, 0\\
A'\; ...\; A\; ...\; B
}
\hspace{20pt}
\vline\quad
\Longstack[l]{
r=0\\
\;\,\ ...\\
\;\,\ ...\\
r=2^{\ceil{\log_2 n}}
}
\end{center}
where $A'\; ...\; A\; ...\; B\x 0\x ...\x 0$ denotes the string with $2\cdot2^{\ceil{\log_2 n}}$ entries where we have appended $2^{\ceil{\log_2 n}}$ elements in "0". Thus, $A\ ...\ B$ reappears after $2^{\ceil{\log_2 n}}$ steps, implying it has a period of at most $2^{\ceil{\log_2 n}}$. As after $2^{\ceil{\log_2 n}}$ steps $A\ ...\ B$ must be mapped back to itself, the period must be a factor of $2^{\ceil{\log_2 n}}$, which only has prime factor $2$. Therefore, the period of $A\ ...\ B$ must itself be a power of 2.
\end{proof}

We now have the tools to understand how purity decays with circuit depth under $\mathrm{CNOT}_n$ with local depolarizing noise. We will consider the circuit in Figure \ref{fig:ladder_heisenberg}.

\begin{lemma}\label{lem:conc-ladder-purity}
For any computational basis state $\dya{z_n}$ with $z_n \in \{0,1\}^n$, we have
\begin{align}
    P\Big(\WC^{(n_d)\dag}_L(\dya{z_n})\Big) - \frac{1}{2^n} &\leq (1-p)^{2\Gamma_{n,n_d,L}}\left(P\big(\dya{z}\big) - \frac{1}{2^n}\right) = (1-p)^{2\Gamma_{n,n_d,L}}\,, \label{eq:conc-purity}
\end{align}
with
\begin{align}
    \Gamma_{n,n_d,L} =
    \begin{cases}
    \left\lfloor L\cdot \frac{n_d }{2^{\floor{\log_2 n}}} \right\rfloor\,, & \text{for}\ n_d = 1 \,,\\
    \left\lfloor L\cdot \frac{n_d }{2^{\ceil{\log_2 n}}} \right\rfloor\,, & \text{for}\ n_d > 1\,,
    \end{cases} 
\end{align}
where $P(\cdot)$ denotes the purity function, 
and $\floor{\cdot}$ and $\ceil{\cdot}$ denote the integer floor and ceiling functions respectively.
\end{lemma}

\begin{remark}
If $n$ is a power of 2, then the exponent in Eqs.~\eqref{eq:conc-purity} 
is given 
\begin{equation}
    \Gamma_{n,n_d,L} = \left\lfloor L\cdot \frac{n_d}{n} \right\rfloor\,.
\end{equation}
Thus, we see that the exponent in Eqs.~\eqref{eq:conc-purity} is approximately equal to $L\cdot \frac{n_d}{n}$.
\end{remark}

\begin{proof}
We start with the proof for number of dirty qubits $n_d=1$, before showing how a looser bound can be constructed for arbitrary $n_d$. From Lemma \ref{prop:maxcycle} we know that the period of a string of length $n$ is at most $2^{\ceil{\log_2 n}}$. We now show that there always exists a binary string with period $2^{\floor{\log_2 n}}$ that corresponds to the Pauli string that is least affected by noise in the circuit in Fig.~\ref{fig:ladder_heisenberg}.

Consider a string of length $n$ made up of all "$0\,$"\,s, with "$1$" on the $2^{\floor{\log_2 n}}th$ entry. By inspection, we see that the "$1$" propagates out to the left and the right in a lightcone structure. In addition, due to Lemma \ref{prop:maxcycle}, the string must repeat after $2^{\floor{\log_2 n}}$ steps as

\begin{center}
\Longstack{
0\x ...\y ... \x 0\x 1\x 0\x ...\x ... \x  0 \x \hphantom{0 \x 0 \x 0}\\
\,\,\y\z 0\x ...\x 0\x 1\x 1\x 0 \x \hphantom{0}\x \hphantom{...} \x \ 0\, \y \hphantom{0} \x \hphantom{0} \x\z \\
\,\hphantom{0}\z \rotatebox[]{-50}{$\rightarrow$}\x \rotatebox[]{50}{$\leftarrow$} \z \rotatebox[]{50}{$\leftarrow$} \hphantom{0} \hspace{3pt} ... \hspace{3pt} \, \hphantom{0} \rotatebox[]{-50}{$\rightarrow$} \z \rotatebox[]{-50}{$\rightarrow$} \z \hphantom{0} \x \ \hphantom{0}  \hphantom{0} \x \rotatebox[]{-50}{$\rightarrow$} \x\x \hphantom{0} \\
\hphantom{0} \hphantom{0} \x 0\x 1\y ... \x ...   \y 1 \x 0 \x \hphantom{...} \x \hphantom{0} \x\ 0 \x \hphantom{0} \\
\hphantom{0}\x \hphantom{0}\x \y \hspace{1pt} 1 \hspace{29.5pt} ... \hspace{30.5pt} 1 \x 0 \x \,... \x ... \x 0 \x \hphantom{0}\,\ \\
\hphantom{\x {0} \x {0} \x {0}} 0\y ...\x ... \y 0\x 1\x 0\x\, ...\x ... \x 0 \y
}
\hspace{10pt}
\vline\quad
\Longstack[l]{
r=0\\
r=1\\
\;\,\ ...\\
r=2^{\floor{\log_2 n}}-2\\
r=2^{\floor{\log_2 n}}-1\\
r=2^{\floor{\log_2 n}}
}
\end{center}

Due to the lightcone structure, it is clear there is only one string per cycle with "$1$" as its first entry. Thus, every $2^{\floor{\log_2 n}}$ steps there is one string with a "$1$" in the first entry. Now we show that no other string that produces a cycle with proportionally fewer "$1$"\,s in the first entry every $L$ steps. 

First, we remark that no nontrivial string with period smaller than $2^{\floor{\log_2 n}}$ can produce such a cycle, as due to the lightcone structure of "$1$"\,s propagating left, there must always be at least one string with a "$1$" in the first entry. Second, we note that if $n$ is a power of 2, then $2^{\floor{\log_2 n}}=2^{\ceil{\log_2 n}}$ and thus with the above we have already found the cycle with the proportionally fewest "$1$"s in the first entry.

Thus, it only remains to consider strings with length $n$ not a power of 2, and with period  $2^{\ceil{\log_2 n}}$ (as this is the maximum possible period, as set by Lemma \ref{prop:maxcycle}). We remark that for such strings, there must be at least one string with "$1$" as the first entry within each cycle, as any "$1$" propagates left in a lightcone. Without loss of generality, we consider the first string of each cycle to be such a string. As we are considering non-trivial strings and cycles, the second string in the cycle must also have a "$1$" at some position. If the left-most "$1$" is the first entry, then we have trivially found a cycle with two adjacent strings with "$1$" in the first entry. If the left-most "$1$" is not the first entry, the "$1$" will propagate left with each step in a light cone. As $2^{\ceil{\log_2 n}}>n$ if $n$ is not a power of 2, it is assured that the light cone of "$1$"\,s reaches the first entry before the end of the cycle. Thus, any string with with period  $2^{\ceil{\log_2 n}}$ has at least two strings per cycle with "$1$" in the first entry.

We can summarize the above result as 
\begin{equation}\label{eq:1s-entry1}
    \text{number of "1"\,s in first entry after $L$ steps} \geq \floor{L/2^{\floor{\log_2 n}}}\,,
\end{equation}
for any string in $\{\id, Z\}^{\otimes n}$, where the bound is achievable with the string of of all "$0\,$"\,s, with "$1$" on the $2^{\floor{\log_2 n}}th$ entry.

We now return to the quantum problem, recalling that $Z$ corresponds to "$1$" and $\id$ corresponds to "$0$". From Eq.~\eqref{eq:1s-entry1} we can conclude that through the $L$ channel instances that make up $\WC_L^{(1)}$, any input string in $\{Z,\id\}^{\otimes n}$ will produce at least $\floor{L/2^{\floor{\log_2 n}}}$ strings with $Z$ in the first register. Thus, for any $\sigma \in \{Z, {\id}\}^{\otimes n}/\{{\id}^{\otimes n}\}$ we have
\begin{equation}\label{eq:ladder-conc}
    \WC^{(1)\dag}_L(\sigma) = q^{(1)}_{\sigma,L} \sigma'\,,
\end{equation}
where $\sigma' \in \{Z, {\id}\}^{\otimes n}/\{{\id}^{\otimes n}\}$ and for some $\{q^{(1)}_{\sigma,L}\}_\sigma$ satisfying
\begin{equation}\label{eq:q-1noise}
    \max_{\sigma \in \{Z, {\id}\}^{\otimes n}/\{{\id}^{\otimes n}\}} |q^{(1)}_{\sigma,L}| = (1-p)^{\floor{L/2^{\floor{\log_2 n}}}}\,.
\end{equation}
Moreover, the map in Eq.~\ref{eq:ladder-conc} bijectively maps Pauli strings to Pauli strings. We note that computational basis states can be written as $\dya{z_n} = \frac{1}{2^n} \sum_{\sigma \in \{\id, Z\}^{\otimes n}} s_{z_n,\sigma} \sigma$ where $s_{z_n,\sigma} \in \{+1,-1\}$ and $s_{z_n,\id} =1$.
We can write
\begin{align}\label{eq:2renyi2}
    P\!\left(\WC^{(1)\dag}_L(\dya{z_n}) \right) - \frac{1}{2^n} &= \bigg\| \frac{1}{2^n} \sum_{\sigma \in \{\id, Z\}^{\otimes n}/\{\id^{\otimes n}\}} s_{z_n,\sigma} \WC^{(1)\dag}_L(\sigma) \bigg\|^2_2 \\
    &= \bigg\| \frac{1}{2^n} \sum_{\sigma \in \{\id, Z\}^{\otimes n}/\{\id^{\otimes n}\}} s_{z_n,\sigma} q^{(1)}_{\sigma,L}\sigma' \bigg\|^2_2 \\
    &= \frac{1}{2^{2n}} \sum_{\sigma \in \{\id, Z\}^{\otimes n}/\{\id^{\otimes n}\}} \left|s_{z_n,\sigma} q^{(1)}_{\sigma,L} \right|^2 \Tr[\id] \\
    &= \frac{1}{2^n} \sum_{\sigma \in \{\id, Z\}^{\otimes n}/\{\id^{\otimes n}\}} \left| q^{(1)}_{\sigma,L} \right|^2 \\
    &\leq  (1-p)^{2\floor{L/2^{\floor{\log_2 n}}}} \left(\frac{2^n-1}{2^n}\right) \\
    &= (1-p)^{2\floor{L/2^{\floor{\log_2 n}}}} \left( P\!\left(\dya{z_n} \right) - \frac{1}{2^n} \right)\,,
\end{align}
where in the first equality we have used Lemma \ref{lem:purity}, in the second equality we have used Eq.~\eqref{eq:ladder-conc}, in the third equality we have used the definition of the Schatten 2-norm and the orthogonality of the Pauli matrices under the Hilbert-Schmidt inner product. The inequality comes from Eq.~\ref{eq:q-1noise}.

In order to generalize the results from $\WC^{(1)\dag}_L$ to $\WC^{(n_d)\dag}_L$ (that is, one dirty qubit to $n_d$ dirty qubits), we argue that for each extra dirty qubit, any string $\{Z, \id\}^{\otimes n}$ generates a cycle that has at least one extra $Z$ per cycle in the first $n_d$ qubits. We now relax our previous result and consider cycles of maximal length $2^{\floor{\log_2 n}}$. As argued above, any non-trivial string generates a cycle such that at least one string with "$1$" as the first entry. We suppose that such a string appears in row $k$. From the rules of a Pascal triangle, this implies the must be a "$1$" in row $k-1$ either as the first entry or the second, i.e. we have either
\begin{center}
\Longstack{
0\x 1\\
1
} ,\quad or \quad
\Longstack{
1\x 0\\
1
}\quad\quad\vline\quad\quad
\Longstack[l]{
r=k-1\\
r=k\,.
}
\end{center}
We recall that this corresponds to a Pauli string with a $Z$ in the first or second qubit. Thus, if the first and second qubits are now dirty, this guarantees for any input string an additional factor of $1-p$ per cycle. In general we can consider $n_d\leq n$ dirty qubits and the light cone of influence up from the first entry of row $k$ up to the first $n_d$ entries of row $k-n_d-1$. In this triangular light cone, there will exist at least $n_d$ "$1$"s. Thus, we have 
\begin{equation}
    \text{number of "1"\,s in first $n_d$ entries after $L$ steps} \geq \floor{ L n_d/2^{\ceil{\log_2 n}}}\,.
\end{equation}
This implies that 
\begin{equation}\label{eq:ladder-conc-2}
    \WC^{(n_d)\dag}_L(\sigma) = q^{(n_d)}_{\sigma,L} \sigma'\,,
\end{equation}
where $\sigma' \in \{Z, {\id}\}^{\otimes n}/\{{\id}^{\otimes n}\}$ and for some $\{q^{(1)}_{\sigma,L}\}_\sigma$ satisfying
\begin{equation}\label{eq:q-1noise-2}
    \max_{\sigma \in \{Z, {\id}\}^{\otimes n}/\{{\id}^{\otimes n}\}} |q^{(n_d)}_{\sigma,L}| = (1-p)^{\floor{n_d L/2^{\floor{\log_2 n}}}}\,.
\end{equation}
Proceeding with the same steps as the above, for $n_d$ dirty qubits we obtain the desired result with exponent $\left\lfloor L\cdot \frac{n_d }{2^{\ceil{\log_2 n}}} \right\rfloor$.
\end{proof}

We now restate and prove Proposition \ref{prop:conc-ladder-cost}, which shows how the output distribution of Figure \ref{fig:ladder_heisenberg}(a) exponentially concentrates.

\begin{proposition}\label{prop:appdx-conc-ladder-cost}
Consider the circuit in Figure \ref{fig:ladder_heisenberg}(a) with $n_d$ dirty qubits and computational basis measurement $O_{z_n}=\dya{z_n}$; $z_n \in \{0,1\}^n$. Denote the channel that describes the action of the gates and noise in the circuit as $\WC^{(n_d)}_L$. For any input state $\rho$, the expectation value concentrates as
\begin{equation}
    \left|\Tr[\WC^{(n_d)}_L(\rho)O_{z_n}] - \frac{1}{2^n}\right| \leq  (1-p)^{\Gamma_{n,n_d,L}}\sqrt{P(\rho)}\,,
\end{equation}
for any computational basis measurement, where $P(\rho)$ denotes the purity of $\rho$ and we have defined
\begin{align}
    \Gamma_{n,n_d,L} =
    \begin{cases}
    \left\lfloor L\cdot \frac{n_d }{2^{\floor{\log_2 n}}} \right\rfloor\,, & \text{for}\ n_d = 1 \,,\\
    \left\lfloor L\cdot \frac{n_d }{2^{\ceil{\log_2 n}}} \right\rfloor\,, & \text{for}\ n_d > 1\,.
    \end{cases} 
\end{align}
\end{proposition}

\begin{proof}
We have 
\begin{align}
    \left|\Tr[\WC^{(n_d)}_L(\rho)O_{z_n}] - \frac{1}{2^n}\right| &= \left|\Tr[\WC^{(n_d)}_L(\rho)O_{z_n}] - \Tr\Big[\WC^{(n_d)}_L(\rho)\frac{\id}{2^n}\Big] \right|\\
    &= \left|\Tr\Big[\WC^{(n_d)}_L(\rho)\Big(O_{z_n}-\frac{\id}{2^n}\Big)\Big]\right|\\
    &\leq \|\rho\|_2\, \Big\|\big(\WC^{(n_d)}_L\big)^{\dag}\Big(O_{z_n}-\frac{\id}{2^n}\Big)\Big\|_2 \\
    &= \sqrt{P(\rho)}\sqrt{P\Big(\big(\WC^{(n_d)}_L\big)^{\dag}(\dya{z_n})\Big)-\frac{1}{2^n}}\\
    &\leq \sqrt{P(\rho)} (1-p)^{\Gamma_{n,n_d,L}}\,,
\end{align}
where in the first equality we have used the fact that $\Tr[\WC^{(n_d)}_L(\rho)]=1$, the second equality is simply a grouping of terms, the first inequality is due to the tracial matrix H{\"o}lder's inequality \cite{baumgartner2011inequality}, the third equality is an application of Lemma \ref{lem:purity}, and the final inequality is an application of Lemma \ref{lem:conc-ladder-purity}.
\end{proof}

\subsubsection{Pauli observables} \label{sec:pauli-observables}

We remark that Eq.~\ref{eq:ladder-conc-2} also implies the concentration of Pauli observables with ladder circuits. We make this precise with the following proposition.

\begin{supproposition}[Pauli observables]\label{prop:pauli-observables}
Consider the reverse ladder circuit $\WC^{(n_d)\dag}_L$ (pictured in Fig.~\ref{fig:ladder_heisenberg}b ) with computational basis input state $\dya{z_n}$; $z_n \in \{0,1\}^n$ and measurement observable $O$. Suppose further that $O$ can be decomposed in the Pauli basis as $O=\sum_{\sigma \in \mathrm{W}} \omega_{\sigma} \sigma$ with $\mathrm{W} \subset \{\id, X, Y, Z \}^{\otimes n}$. Then, the output of this circuit concentrates as
\begin{align}
    \left| \Tr\left[\WC^{(n_d)\dag}_L(\dya{z_n}) O\right] - \frac{1}{2^n} \Tr\left[ O\right] \right| \leq |\mathrm{W}| \|\vec{\omega}\|_{\infty} (1-p)^{\floor{n_d L/2^{\floor{\log_2 n}}}} \,,
\end{align}
where we denote $\|\vec{\omega}\|_{\infty}:=\max_{\sigma} \{\omega_{\sigma} \}_{\sigma}$ as the Pauli coefficient of largest magnitude. Thus, the output of the circuit concentrates to the maximally mixed value exponentially quickly with increasing circuit depth. 
\end{supproposition} 

\begin{proof}
Starting from Eq.~\ref{eq:ladder-conc-2} we can directly evaluate
\begin{align}
    \left| \Tr\left[\WC^{(n_d)\dag}_L(\dya{z_n}) O\right] - \frac{1}{2^n} \Tr\left[ O\right] \right| &= \frac{1}{2^n} \sum_{\sigma \in \{\id, Z\}^{\otimes n}/\{\id^{\otimes n}\}} \sum_{\sigma'' \in \mathrm{W}} \left| \Tr\left[q^{(n_d)}_{\sigma,L} \sigma' \omega_{\sigma''} \sigma'' \right] \right| \\
    &= \frac{1}{2^n} \sum_{\sigma \in \{\id, Z\}^{\otimes n}/\{\id^{\otimes n}\}} \sum_{\sigma'' \in \mathrm{W}} \left| q^{(n_d)}_{\sigma,L} \omega_{\sigma''} \right| \delta_{\sigma'  \sigma''} 2^n \\
    &\leq \frac{1}{2^n} \sum_{\sigma \in \{\id, Z\}^{\otimes n}/\{\id^{\otimes n}\}} \sum_{\sigma'' \in \mathrm{W}} (1-p)^{\floor{n_d L/2^{\floor{\log_2 n}}}} \|\vec{\omega}\|_{\infty} \delta_{\sigma'  \sigma''} 2^n \\
    &\leq |\mathrm{W}| \|\vec{\omega}\|_{\infty} (1-p)^{\floor{n_d L/2^{\floor{\log_2 n}}}}\,,
\end{align}

where in the first line we have decomposed $\dya{z_n}$ and $O$ into the Pauli basis, the second line comes from the identity $\Tr[\sigma' \sigma''] = \delta_{\sigma' \sigma''} 2^n$ where $\delta_{\sigma' \sigma''}$ is the Dirac delta function, the third line comes from taking the maximal magnitudes of $q^{(n_d)}_{\sigma,L}$ and $\omega_{\sigma''}$ separately, and the final line comes by noting there can only be at most $|\mathrm{W}|$ non-zero terms in the sum. 
\end{proof}

We remark that this bound is particularly relevant if the measurement operator is composed of a small number of Pauli operators with bounded coefficients, compared to the circuit depth. For instance, if $\|\vec{\omega}\|_{\infty}:=\max_{\sigma} \{\omega_{\sigma} \}_{\sigma} \in \OC(\mathrm{poly}(n))$ and $N_O \in \OC(\mathrm{poly}(n))$ but $L=\Omega(n)$, then Supplemental Proposition \ref{prop:pauli-observables} implies exponential concentration with increasing system size.

\subsection{Model 2 - entangled registers}\label{sec:entangleUV}

In this section we consider circuits of the form in Figure \ref{fig:entangle_UV_heisenberg}(a). In order to analyze such circuits we bound the output purity of the reverse circuit with computational basis input states (Figure \ref{fig:entangle_UV_heisenberg}(b)). This will allow us to consider the circuit in Figure \ref{fig:entangle_UV_heisenberg}(a) using the Heisenberg picture. 

\begin{figure}[t]
	\includegraphics[width= .9 \columnwidth]{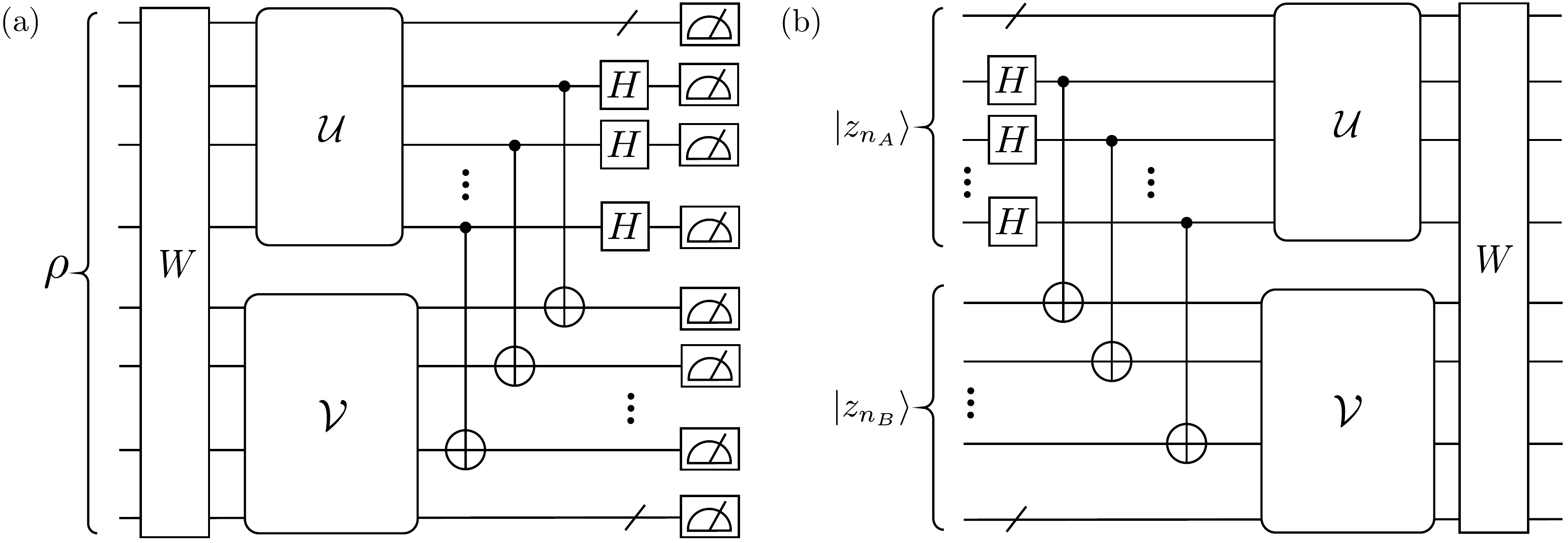}
	\caption{\label{fig:entangle_UV_heisenberg} \textbf{}(a) We consider a class of circuits with entangled measurements, that includes as a special case the Hilbert-Schmidt Test circuit. (b) In order to analyze such circuits, we bound the output purity of the reverse circuit with computational basis input states. }
\end{figure}

We consider two distinct subsystems $A$ and $B$, of size $n_A$ qubits and $n_B$ qubits respectively. The input to the circuit in Figure \ref{fig:entangle_UV_heisenberg}(b) is a computational basis state $\ket{z_n} = \ket{z_{n_A}}\ket{z_{n_B}}$; $z_n \in \{0,1\}^n$. We first implement $m$ entangling operations with Hadamard gates and CNOTs on a subset of $2m$ registers between $A$ and $B$, with an operation we denote as $\BC_m$. Note that we allow for any generic $m$ satisfying $m>0$. The second step of the circuit is to act on $A$ with unitary evolution $\UC$, and act on $B$ with unitary $\VC$. Finally, we allow the action of a global unitary $\WC$ on the joint system $AB$. We remark that this circuit includes as a special case the Hilbert-Schmidt Test circuit \cite{khatri2019quantum}. 

We consider three modifications of this circuit where we now consider the effects of noise. In the first modification, we consider subsystem A to be noisy, and thus we modify the channel $\UC\rightarrow\widetilde{\UC}$, where $\widetilde{\UC}$ is the channel corresponding to the hardware implementation of the unitary channel $\UC$ with an instance of local depolarizing noise $\DC_p^{\otimes n_A}$  in between each layer of gates. This follows the noise model as considered in Ref.~\cite{wang2020noise}, and we present a schematic of this modification in Figure \ref{fig:noisy_unitary}. In the second modification, we consider $\VC\rightarrow\widetilde{\VC}$, where $\widetilde{\VC}$ is the noisy version of the unitary channel $\VC$. We define the depth of the unitaries as the number of unitary layers in such an implementation and denote the depth of $\widetilde{\UC}$($\widetilde{\VC}$) as $L_{\UC}$($L_{\VC}$). Finally, we also consider the scenario when both $A$ and $B$ are noisy, and we have $\UC\rightarrow\widetilde{\UC}, \VC\rightarrow\widetilde{\VC}$.

\begin{figure}[t]
    \begin{center}
	\includegraphics[width= .7 \columnwidth]{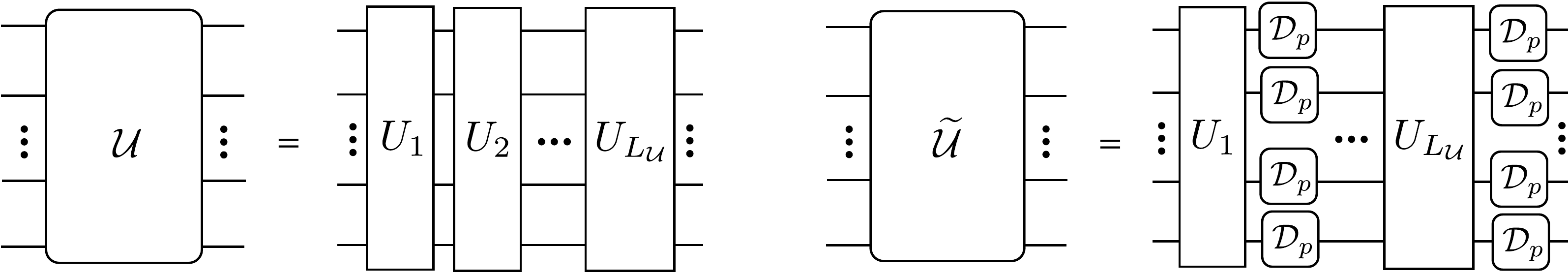}
    \end{center}
\caption{\label{fig:noisy_unitary} \textbf{} We consider the depth of $\UC$ to be the number of non-parallelizable unitary layers in a given hardware implementation. When we consider a noisy modification $\UC \rightarrow \widetilde{\UC}$, we insert layers of local depolarizing noise with depolarizing probability $p$ in between each unitary layer.}
\end{figure}

\begin{lemma}\label{lem:entangleUV-purity}
Consider the circuit in Figure \ref{fig:entangle_UV_heisenberg}(b) with entangling operation $\BC_m$, separable noisy unitary evolution $\TC$ and global unitary evolution $\WC$. The purity of the circuit concentrates as
\begin{align}\label{eq:L_T}
    P\!\left(\WC\circ\TC\circ\BC_{m}(\dya{z_n}) \right) - \frac{1}{2^n}
    \leq (1-p)^{L_\TC} \left(P\!\left(\dya{z_n} \right) - \frac{1}{2^n}\right)\,,
\end{align}
for any computational basis state input $\dya{z_n};$ $z_n \in \{0,1\}^n$, where we have defined
\begin{align}
    L_\TC = 
    \begin{cases}
    {L_{\UC}}\,,& \text{for}\ \TC = \widetilde{\UC} \otimes \VC\,, \\
    {L_{\VC}}\,,& \text{for}\ \TC = {\UC} \otimes \widetilde{\VC}\,, \\
    {L_{\UC}+L_{\VC}}\,,& \text{for}\ \TC = \widetilde{\UC} \otimes \widetilde{\VC} \,.
    \end{cases}     
\end{align}
\end{lemma}

Thus, when only subsystem $A$ ($B$) is noisy, there is still exponential scaling, however, compared to the fully noisy setting the exponent is modified by factor $\frac{L_{\UC}}{L_{\UC}+ L_{\VC}}$ ($\frac{L_{\VC}}{L_{\UC}+ L_{\VC}}$). If we consider the unitaries to have linear depth in the number of qubits and scaling with the same proportionally factor (i.e. $L_{\UC}/n_A = L_{\VC}/n_B$), then the modification to the exponent is $\frac{n_A}{n_A+n_B}$ ($\frac{n_B}{n_A+n_B}$).

\begin{proof}
Consider the mapping of strings in $\{ \id, Z\}^{\otimes 2}$ under a Hadamard gate followed by a CNOT. For map $\BC(.)= (H \otimes \id)CNOT(.)CNOT^\dag(H \otimes \id)$ we have
\begin{align}
    \id\id &\xrightarrow{\BC} \id\id\,, \label{eq:Bmap1}\\
    \id Z &\xrightarrow{\BC} ZZ\,, \\
    Z\id &\xrightarrow{\BC} XX\,, \\
    ZZ &\xrightarrow{\BC} -YY\,. \label{eq:Bmap4}
\end{align}
This implies that after the CNOTs in the circuit in Figure \ref{fig:entangle_UV_heisenberg}(b), the strings are (up to a phase factor) exclusively of the form $S_A\otimes S_B\in\mathit{Bell}$, where $S_A = \{ \id, Z\}^{\otimes n_{A}-m} \otimes A \cup\{\id^{\otimes n_A}\}$ and $S_B \in A\otimes \{ \id, Z\}^{\otimes n_{B}-m}\cup\{\id^{\otimes n_B}\}$, with $A\in\{X,Y,Z \}^{\otimes m}$. In particular, we make a key observation that $S_A = \id^{\otimes n_A}$ if and only if $S_B = \id^{\otimes n_B}$. This implies that we can import over the results of Ref.~\cite{wang2020noise} and treat subsystems $A$ and $B$ individually in our analysis of the evolution under $\TC$. Namely, we have for any $\sigma_{A} \in P_{n_A}$, with $P_{n_A}$ the group of Pauli operators acting on subsystem $A$, we have
\begin{align}\label{eq:noise-on-A}
    \DC_p^{\otimes n_A}(\sigma_{A}) = q_{\sigma_{A}} \sigma_{A}\,, 
\end{align}
\begin{align}
    \DC_p^{\otimes n_B}(\sigma_{B}) = q_{\sigma_{B}} \sigma_{B}\,,
\end{align}
where $|q_{\sigma_B}|\leq 1-p$.
A consequence of this is that for any unitary channel $\YC$ and for input operator $\sum_{\sigma_{A} \in P_{n_A}, \sigma_B \in P_{n_B}} \lambda(\sigma_{A},\sigma_{B}) \sigma_{A}\otimes \sigma_{B}$ where $\lambda(\sigma_{A},\sigma_{B}) \in \mathbb{C}$ are arbitrary coefficients, we have
\begin{align}
    \bigg\|\YC\circ(\DC_p^{\otimes n_A}\otimes\IC)\Big(\sum_{\sigma_{A} \in P_{n_A}, \sigma_B \in P_{n_B}} \lambda(\sigma_{A},\sigma_{B}) \sigma_{A}\otimes \sigma_{B}\Big)\bigg\|_2
    &= \bigg\| \YC\Big(\sum_{\sigma_{A} \in P_{n_A}, \sigma_B \in P_{n_B}}  q_{\sigma_A}\lambda(\sigma_{A},\sigma_{B}) \sigma_{A}\otimes \sigma_{B}\Big)\bigg\|_2 \label{eq:bell-onestep}\\
    &= \bigg\|\sum_{\sigma_{A} \in P_{n_A}, \sigma_B \in P_{n_B}} q_{\sigma_A}\lambda(\sigma_{A},\sigma_{B}) \sigma_{A}\otimes \sigma_{B}\bigg\|_2 \label{eq:bell-onestep2}\\
    &= \left( \sum_{\sigma_{A} \in P_{n_A}, \sigma_B \in P_{n_B}} |q_{\sigma_A}\lambda(\sigma_{A},\sigma_{B})|^2 \Tr[(\sigma_{A}\otimes \sigma_{B})^2]\right)^{1/2}\\
    &\leq \left( \sum_{\sigma_{A} \in P_{n_A}, \sigma_B \in P_{n_B}} (1-p)^2|\lambda(\sigma_{A},\sigma_{B})|^2 \Tr[(\sigma_{A}\otimes \sigma_{B})^2]\right)^{1/2}\\
    &= (1-p) \bigg\|\sum_{\sigma_{A} \in P_{n_A}, \sigma_B \in P_{n_B}} \lambda(\sigma_{A},\sigma_{B}) \sigma_{A}\otimes \sigma_{B}\bigg\|_2\,,
\end{align}
where in the first equality we have used Eq.~\eqref{eq:noise-on-A}, in the second equality we have used the unitary invariance of Schatten norms, the third equality comes from the definition of the Schatten 2-norm, and the inequality comes from the bound $|q_{\sigma_A}|\leq 1-p$. The final equality again comes from the definition of the Schatten 2-norm.
We note that the output of $S_A\otimes S_B \in \textit{Bell}/\{\id^{\otimes n}\}$ under unital channels separable across the cut $A|B$ can generically be written in the form as the argument in the left-hand side of Eq.~\eqref{eq:bell-onestep}. As we only consider such unitary channels in our analysis of $\TC$, we thus can iteratively apply the above to write the concentration inequality
\begin{align}
    \bigg\|\widetilde{\UC}\otimes\IC\Big(\sum_{S_A \otimes S_B \in \textit{Bell}/\{\id^{\otimes n}\}}& \lambda(S_{A},S_{B}) S_{A}\otimes S_{B}\Big)\bigg\|_2 = \\
    &= \bigg\|\DC_p^{\otimes n_A}\circ\UC_L\circ\DC_p^{\otimes n_A}\circ...\circ\DC_p^{\otimes n_A}\circ\UC_1\otimes\IC\Big(\sum_{S_A \otimes S_B \in \textit{Bell}/\{\id^{\otimes n}\} } \lambda(S_{A},S_{B}) S_{A}\otimes S_{B}\Big)\bigg\|_2 \\
    &\leq (1-p)^{L_{\UC}} \bigg\|\sum_{S_A \otimes S_B \in \textit{Bell}/\{\id^{\otimes n}\} } \lambda(S_{A},S_{B}) S_{A}\otimes S_{B}\bigg\|_2\,,\label{eq:conc_UI}
\end{align}
where $\lambda(S_{A},S_{B})\in \mathbb{C}$ are arbitrary coefficients.
By repeating analogous steps to Eqs.~\eqref{eq:bell-onestep}-\eqref{eq:conc_UI} we can also obtain
\begin{equation}
    \bigg\|\IC\otimes\widetilde{\VC}\Big(\sum_{S_A \otimes S_B \in \textit{Bell}/\{\id^{\otimes n}\} } \lambda(S_{A},S_{B}) S_{A}\otimes S_{B}\Big)\bigg\|_2 \leq (1-p)^{L_{\VC}} \bigg\|\sum_{S_A \otimes S_B \in \textit{Bell}/\{ \id^{\otimes n}\} } \lambda(S_{A},S_{B}) S_{A}\otimes S_{B}\bigg\|_2\,.\label{eq:conc_IV}
\end{equation}

We now have the tools to prove the lemma. We recall that $\dya{z_n} = \frac{1}{2^n} \sum_{\sigma \in \{\id, Z\}^{\otimes n}} s_{z_n,\sigma} \sigma$ where $s_{z_n,\sigma} \in \{+1,-1\}$ and $s_{z_n,\id} =1$. Then, we can write 
\begin{align}
    P\!\left(\WC\circ\TC\circ\BC_{m}(\dya{z_n}) \right) - \frac{1}{2^n}
    &= \left\| \WC\circ\TC\circ\BC_{m}\left(\dya{z_n} - \frac{\id^{\otimes n}}{2^n}\right) \right\|_2^2\\
    &= \left\| \WC\circ\TC\Big( \frac{1}{2^n}\sum_{S_A \otimes S_B \in \textit{Bell}/\{\id^{\otimes n}\} } (-1)^{f(z_n,S_{A},S_{B})} S_{A}\otimes S_{B}\Big) \right\|_2^2\\
    & \leq (1-p)^{L_\TC} \left\|  \frac{1}{2^n}\sum_{S_A \otimes S_B \in \textit{Bell}/\{\id^{\otimes n}\} } (-1)^{f(z_n,S_{A},S_{B})} S_{A}\otimes S_{B} \right\|_2^2\\
    &= (1-p)^{L_\TC} \left\|  \BC_{m}\left(\dya{z_n} - \frac{\id^{\otimes n}}{2^n}\right) \right\|_2^2 \\
    &= (1-p)^{L_\TC} \left\|  \dya{z_n} - \frac{\id^{\otimes n}}{2^n} \right\|_2^2 \\
    &= (1-p)^{L_\TC} \left(P\!\left(\dya{z_n} \right) - \frac{1}{2^n}\right)\,,
\end{align}
where $L_{\TC}$ is defined in Eq.~\eqref{eq:L_T} and $f(z_n, S_A, S_B)$ is a function that generates the appropriate phase. In the first equality we have used Lemma \ref{lem:purity}, in the second equality we have used Eqs.~\eqref{eq:Bmap1}-\eqref{eq:Bmap4}. The inequality comes from Eqs.~\eqref{eq:conc_UI} and \eqref{eq:conc_IV}, and the unitary invariance of the Schatten norms. In the last three equalities, we again use Eqs.~\eqref{eq:Bmap1}-\eqref{eq:Bmap4}, the unitary invariance of the Schatten norms, and Lemma \ref{lem:purity}.
\end{proof}

Using Lemma \ref{lem:entangleUV-purity} we can now prove our proposition on the concentration of output distributions in circuits with entangled measurements.

\begin{proposition}\label{prop:appdx-conc-UV-cost}
Consider the circuit in Figure \ref{fig:entangle_UV_heisenberg}(a) with computational basis measurement $O_{z_n}=\dya{z_n}$; $z_n \in \{0,1\}^n$. For any input state $\rho$ and choice of measurement string $z_n$, the expectation value concentrates as
\begin{equation}
    \Tr\big[\BC_m^\dag\circ\TC\circ\WC(\rho)O_{z_n}\big] - \frac{1}{2^n} \leq  (1-p)^{L_\TC} \sqrt{P(\rho)}\,,
\end{equation}
where $P(\rho)$ is the purity of $\rho$ and 
\begin{align}
    L_\TC = 
    \begin{cases}
    {L_{\UC}}\,,& \text{for}\ \TC = \widetilde{\UC} \otimes \VC\,, \\
    {L_{\VC}}\,,& \text{for}\ \TC = {\UC} \otimes \widetilde{\VC}\,, \\
    {L_{\UC}+L_{\VC}}\,,& \text{for}\ \TC = \widetilde{\UC} \otimes \widetilde{\VC} \,.
    \end{cases}     
\end{align}
\end{proposition}

\begin{proof}
In similar steps to the proof of Proposition \ref{prop:conc-ladder-cost}
\begin{align}
    \Tr\big[\BC_m^\dag\circ\TC\circ\WC(\rho)O_{z_n}\big] - \frac{1}{2^n} &= \Tr\big[\BC_m^\dag\circ\TC\circ\WC(\rho)O_{z_n}\big] - \Tr\Big[\BC_m^\dag\circ\TC\circ\WC(\rho)\frac{\id}{2^n}\Big]\\
    &= \Tr\Big[\BC_m^\dag\circ\TC\circ\WC(\rho)\Big(O_{z_n}-\frac{\id}{2^n}\Big)\Big]\\
    &\leq \|\rho\|_2\, \Big\|\WC^{\dag}\circ\TC^\dag\circ\BC_m\Big(O_{z_n}-\frac{\id}{2^n}\Big)\Big\|_2 \\
    &= \sqrt{P(\rho)}\sqrt{P(\WC^{\dag}\circ\TC^\dag\circ\BC_m(\dya{z_n}))-\frac{1}{2^n}}\\
    &  \leq \sqrt{P(\rho)} (1-p)^{L_{\TC}}\,, 
\end{align}
where in the first equality we have used the fact that $\Tr[\BC_m^\dag\circ\TC\circ\WC(\rho)]=1$, the second equality is simply a grouping of terms, the first inequality is due to the tracial matrix H{\"o}lder's inequality \cite{baumgartner2011inequality}, the third equality is an application of Lemma \ref{lem:purity}, and the final inequality is an application of  Lemma \ref{lem:entangleUV-purity}. 
\end{proof}

\subsection{Gradient scaling}\label{sec:gradientscaling}

In this section, we show how to transport our results to consider gradient scaling. We modify the settings studied in Section \ref{sec:analytics} to include a trainable unitary in the middle of the circuit. In doing so, we constrain ourselves to consider only computational basis input states.

We consider a trainable unitary of the form

\begin{equation} \label{eq:trainable-unitary2}
    Y(\thv)= W_0 \prod_{k=1}^K e^{-i \theta_{k} H_{k}} W_{k}\,,
\end{equation}

where $\{W_k\}_{k=0}^K$ are arbitrary fixed unitary operators and $\{H_k\}_{k=0}^K$ are Hermitian operators. We denote the channel that corresponds to this unitary as $\YC(\thv)$. We modify the circuit in Figure \ref{fig:ladder} to include this trainable unitary, which we present in Figure \ref{fig:ladder_grad}. In the following proposition, we show how partial derivatives of parameters in this unitary exponentially vanish in the circuit depth under a CNOT ladder circuit.

\begin{figure}[t]
    \begin{center}
	\includegraphics[width= .7 \columnwidth]{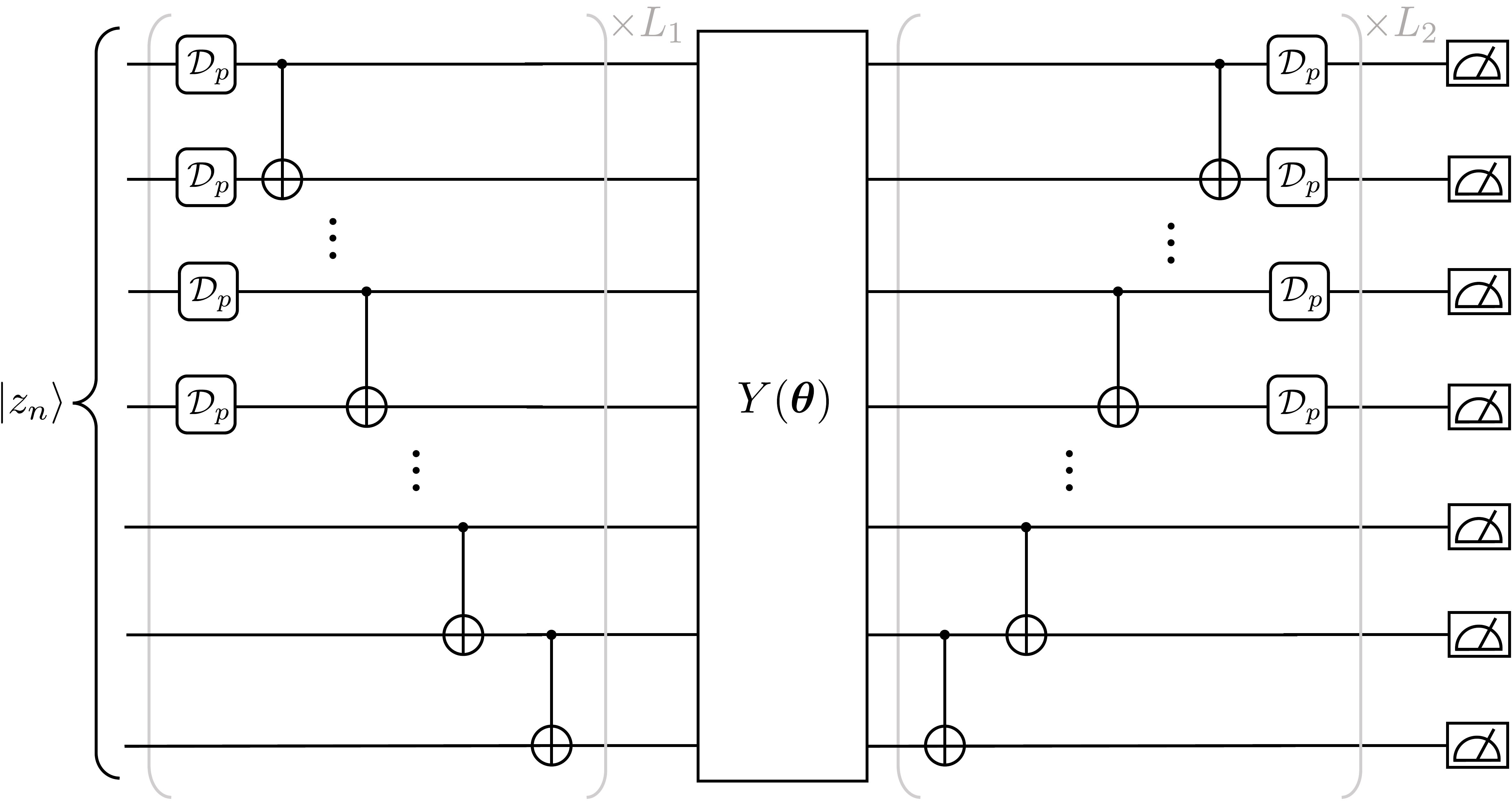}
    \end{center}
	\caption{\label{fig:ladder_grad} Modification to the circuit in Figure \ref{fig:ladder} to include trainable unitary $Y(\thv)$. }
\end{figure}

\begin{proposition}\label{prop:appdx-ladder-grad}
Consider a cost function corresponding to the circuit in Figure \ref{fig:ladder_grad} which we denote as $C(\thv)=\Tr[\WC^{(n_d)}_{L_2}\circ\YC({\thv})\circ(\WC^{(n_d)}_{L_1})^{\dag}(\rho)O_{z_n}]$, where $\rho$ is a computational basis input state, $\WC^{(n_d)}_{L_i}$ denotes the channel corresponding to $L_i$ instances of CNOT ladders and local noise on $n_d$ qubits as defined in Section \ref{sec:ladder}, and with computational basis measurement $O_{z_n}=\dya{z_n}$; $z_n \in \{0,1\}^n$. The partial derivative with respect to parameter $\theta_k$ is bounded as
\begin{equation}\label{eq:ladder-grad}
    |\partial_{\theta_k} C| \leq (1-p)^{\Gamma_{n,n_d,L_1}+\Gamma_{n,n_d,L_2}}\|H_k\|_\infty\,,
\end{equation}
where we denote
\begin{align}
    \Gamma_{n,n_d,L_i} =
    \begin{cases}
    \left\lfloor L_i\cdot \frac{n_d }{2^{\floor{\log_2 n}}} \right\rfloor\,, & \text{for}\ n_d = 1 \,,\\
    \left\lfloor L_i\cdot \frac{n_d }{2^{\ceil{\log_2 n}}} \right\rfloor\,, & \text{for}\ n_d > 1\,,
    \end{cases} 
\end{align}
for $i \in \{1,2\}$.
\end{proposition}

\begin{proof}
We have 
\begin{align}\label{eq:ladder-grad-master}
    |\partial_{\theta_k} C| &= \left|\Tr\left[\WC^{(n_d)}_{L_2}\circ\YC_{-}(\thv)\partial_{\theta_k}\YC_{+}(\thv)\circ(\WC^{(n_d)}_{L_1})^{\dag}(\rho)O_{z_n}\right]\right| \\
    &\leq \left\| \YC^{\dag}_{k-}(\thv)\circ(\WC^{(n_d)}_{L_2})^{\dag}(O_{z_n}) \right\|_2 \left\| \partial_{\theta_k}\YC_{k+}(\thv)\circ(\WC^{(n_d)}_{L_1})^{\dag}(\rho) \right\|_2\,,
\end{align}
where the inequality is due to H\"older's tracial matrix inequality \cite{baumgartner2011inequality}, and we denote $\YC^{\dag}_{k-}(\thv)$ and $\YC^{\dag}_{k+}(\thv)$ as the channels corresponding to unitary operators 
\begin{equation}
    Y_{k-}(\thv) = W_0 \prod_{j=1}^k e^{-i \theta_{j} H_{j}} W_{j}\,,\quad Y_{k+}(\thv) =  \prod_{j=k}^K e^{-i \theta_{j} H_{j}} W_{j}\,,
\end{equation}
respectively. The first term in Eq.~\eqref{eq:ladder-grad-master} can be bounded by considering the unitary invariance of Schatten norms and by using Lemma \ref{lem:conc-ladder-purity} to obtain
\begin{align}
    \| \YC^{\dag}_{-}(\thv)\circ(\WC^{(n_d)}_{L_2})^{\dag}(O_{z_n}) \|_2 &\leq (1-p)^{\Gamma_{n,n_d,L_2}} \sqrt{P(O_{z_n})} \\
    &=(1-p)^{\Gamma_{n,n_d,L_2}}\,, \label{eq:ladder-grad-pt1}
\end{align}
where in the last line we have used the fact that the purity of pure states is equal to $1$. The second term in Eq.~\eqref{eq:ladder-grad-master} can be bounded as 
\begin{align}
    \left\| \partial_{\theta_k}\YC_{k+}(\thv)\circ(\WC^{(n_d)}_{L_1})^{\dag}(\rho) \right\|_2 &=  \left\| \left[H_k,\, \YC_{k+}(\thv)\circ(\WC^{(n_d)}_{L_1})^{\dag}(\rho)\right] \right\|_2 \\
    & \leq \| H_k \|_\infty \left\| \YC_{k+}(\thv)\circ(\WC^{(n_d)}_{L_1})^{\dag}(\rho) \right\|_2 \\
    & \leq \| H_k \|_\infty (1-p)^{\Gamma_{n,n_d,L_1}}\,,\label{eq:ladder-grad-pt2}
\end{align}
where the equality comes by directly evaluating the partial derivative, the first inequality comes from the application of Lemma \ref{lem:commutator}, and the second inequality is an application of Lemma \ref{lem:conc-ladder-purity}. Substituting Eqs.~\eqref{eq:ladder-grad-pt1} and \eqref{eq:ladder-grad-pt2} into Eq.~\eqref{eq:ladder-grad-master} we obtain the gradient scaling result Eq.~\eqref{eq:ladder-grad} as desired.
\end{proof}

Proposition \ref{prop:appdx-ladder-grad} shows that the partial derivatives with respect to any parameter $\theta_k$ exponentially vanishes in the circuit depth. This implies that so long as $\|H_k\|_\infty$ grows at most polynomially in the number of qubits $n$, the circuit in Figure~\ref{fig:ladder_grad} has an NIBP for linear circuit depths.

We can similarly modify the circuit in Figure \ref{fig:entangleUV} to include the trainable unitary $Y(\thv)$ and entangling gates at the start of the circuit. We present this circuit in Figure \ref{fig:entangle_UV_grad}. In the following proposition, we show that partial derivatives of parameters in this unitary also exponentially vanish with increasing depth.

\begin{figure}[t]
    \begin{center}
	\includegraphics[width= .8 \columnwidth]{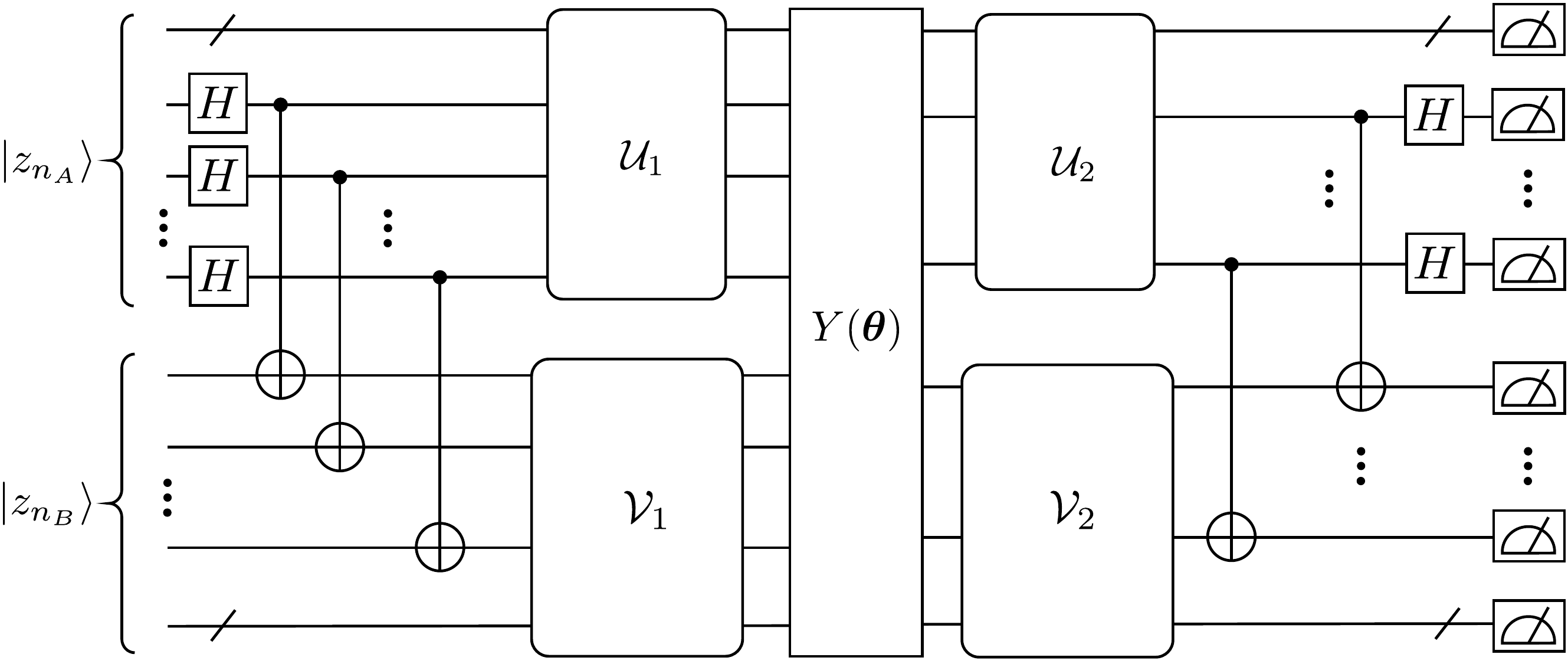}
    \end{center}
	\caption{\label{fig:entangle_UV_grad} Modification to the circuit in Figure \ref{fig:entangleUV} to include trainable unitary $Y(\thv)$. }
\end{figure}

\begin{proposition}
Consider a cost function corresponding to the circuit in Figure \ref{fig:entangle_UV_grad} which we denote as $C(\thv)=\Tr[\BC^{\dag}_{m'}\circ\TC_2\circ\YC(\thv)\circ\TC_1\circ\BC_m(\rho)O_{z_n}]$, where $\rho$ is a computational basis input state, $\BC_{m_i}$ is an entangling operation between ${m_i}$ pairs of qubits in subsystems $A_i$ and $B_i$ as considered in Proposition \ref{prop:conc-UV-cost}, $\TC_i$ denotes noisy unitary evolution separable across the cut $A_i|B_i$, and with computational basis measurement $O_{z_n}=\dya{z_n}$; $z_n \in \{0,1\}^n$. The partial derivative with respect to parameter $\theta_k$ is bounded as
\begin{equation}\label{eq:entangleUV-grad}
    |\partial_{\theta_k} C| \leq (1-p)^{L_{\TC_1}+L_{\TC_2}}\|H_k\|_\infty\,,
\end{equation}
where we denote
\begin{align}
    L_{\TC_i} = 
    \begin{cases}
    {L_{\UC_i}}\,,& \text{for}\ \TC_i = \widetilde{\UC}_i \otimes \VC_i\,, \\
    {L_{\VC_i}}\,,& \text{for}\ \TC_i = {\UC_i} \otimes \widetilde{\VC}_i\,, \\
    {L_{\UC_i}+L_{\VC_i}}\,,& \text{for}\ \TC_i = \widetilde{\UC}_i \otimes \widetilde{\VC}_i \,,
    \end{cases}     
\end{align}
for $i \in \{1,2\}$, and $L_{\UC_i}$ and $L_{\VC_i}$ are defined in the same manner as $L_{\UC}$ and $L_{\VC}$ in Section \ref{sec:entangleUV}.
\end{proposition}

\begin{proof}
Similar to the proof of Proposition \ref{prop:ladder-grad}, we have 
\begin{align}\label{eq:entangleUV-grad-master}
    |\partial_{\theta_k} C| &= \left|\Tr\left[\BC^{\dag}_{m'}\circ\TC_2\circ\YC_{-}(\thv)\partial_{\theta_k}\YC_{+}\circ\TC_1\circ\BC_m(\rho)O_{z_n}\right]\right| \\
    &\leq \left\| \YC^{\dag}_{k-}(\thv)\circ\TC_2^{\dag}\circ\BC_m(O_{z_n}) \right\|_2 \big\| \partial_{\theta_k}\YC_{k+}(\thv)\circ\TC_1\circ\BC_m(\rho) \big\|_2\,,
\end{align}
where the inequality is due to H\"older's tracial matrix inequality \cite{baumgartner2011inequality}. The first term is bounded as 
\begin{align}
    \left\| \YC^{\dag}_{k-}(\thv)\circ(\WC^{(n_d)}_{L_2})^{\dag}(O_{z_n}) \right\|_2 \leq (1-p)^{\TC_2} \,,\label{eq:entangleUV-grad-pt1}
\end{align}
where we have used Lemma \ref{lem:entangleUV-purity}, and the fact that the lemma gives the same result if we consider the adjoint map of $\TC_2$. As in the proof of Proposition \ref{prop:ladder-grad}, the second term in Eq.~\ref{eq:entangleUV-grad-master} is bounded as 
\begin{align}
    \big\| \partial_{\theta_k}\YC_{k+}(\thv)\circ\TC_1\circ\BC_m(\rho) \big\|_2 &=  \left\| \big[H_k,\, \YC_{k+}(\thv)\circ\TC_1\circ\BC_m(\rho)\big] \right\|_2 \\
    & \leq \| H_k \|_\infty \left\| \YC_{k+}(\thv)\circ\TC_1\circ\BC_m(\rho) \right\|_2 \\
    & \leq \| H_k \|_\infty (1-p)^{\TC_1}\,,\label{eq:entangleUV-grad-pt2}
\end{align}
where the equality comes by directly evaluating the partial derivative, the first inequality comes from the application of Lemma \ref{lem:commutator}, and the second inequality is an application of Lemma \ref{lem:entangleUV-purity}. Substituting Eqs.~\eqref{eq:entangleUV-grad-pt1} and \eqref{eq:entangleUV-grad-pt2} into Eq.~\eqref{eq:entangleUV-grad-master} we obtain the gradient scaling result Eq.~\eqref{eq:entangleUV-grad} as required.
\end{proof}

\section{Numerical gradient computation for the HVA ansatz} 
\label{app:shift}
To compute the gradient with respect to  $\theta_{k}$ corresponding to   $e^{-i\theta_{k} H_{XX}}$ in the ansatz~(\ref{eq:HVA}) we introduce "dummy" variables $\tilde{\theta}_{k,1}, \tilde\theta_{k,2}, \dots \tilde\theta_{k,n}$ such that
\begin{equation}
e^{-i\theta_k H_{XX}} = e^{-i\tilde\theta_{k,1} X_1 X_2} e^{-i\tilde\theta_{k,2} X_2 X_3}  \dots e^{-i\tilde\theta_{k,N} X_n X_1}|_{\tilde\theta_{k,1} = \tilde\theta_{k,2} =  \dots = \tilde\theta_{k,n} = \theta_k}.
\end{equation}
We have
\begin{equation}
\partial_{\theta_k}C(\theta_1, \dots,\theta_{k-1},\theta_k, \theta_{k+1}, \dots, \theta_{2L}  ) = \sum_{i=1,\dots,n}\partial_{\tilde\theta_{k,i}}   C(\theta_1, \dots, \theta_{k-1} , \tilde\theta_{k,1}, \tilde\theta_{k,2}, \dots, \tilde\theta_{k,n}, \theta_{k+1}, \dots  ,\theta_{2L})  |_{\tilde\theta_{k,1} = \tilde\theta_{k,2} =  \dots = \tilde\theta_{k,n} = \theta_k }. 
\end{equation}
Using a parameter shift rule for Pauli strings~\cite{schuld2019evaluating} leads to
\begin{equation}
\begin{split}
& \partial_{\tilde\theta_{k,1}}   C(\theta_1, \dots, \theta_{k-1} , \tilde\theta_{k,1}, \tilde\theta_{k,2}, \dots, \tilde\theta_{k,n}, \theta_{k+1}, \dots  ,\theta_{2L})  |_{\tilde\theta_{k,1} = \tilde\theta_{k,2} =  \dots = \tilde\theta_{k,n} = \theta_k} =  \\
 & C(\theta_1, \dots, \theta_{k-1} , \theta_{k}+\pi/4, \theta_{k}, \dots, \theta_{k}, \theta_{k+1}, \dots  ,\theta_{2L}) - C(\theta_1, \dots, \theta_{k-1} , \theta_{k}-\pi/4, \theta_{k}, \dots, \theta_{k}, \theta_{k+1}, \dots  ,\theta_{2L}).  
\end{split}
\end{equation}
Analogously we compute $\tilde\theta_{k,2},\dots,\tilde\theta_{k,n}$. We treat similarly $\partial_{\theta_l}C$ corresponding to  $e^{-i\theta_{l} H_Z}$ in the ansatz~(\ref{eq:HVA}).

\section{Noise models} 
\label{app:noise}
In our numerical studies, we use the local depolarizing  noise model and a noise model of a near-term  trapped-ion quantum computer developed by Trout \textit{et al. }\cite{trout2018simulating}. We assume single-qubit rotations  $R_X(\theta)=e^{-i \theta/2 X} , R_Y(\theta)= e^{-i \theta/2 Y},   R_Z(\theta)= e^{-i \theta/2 Z},$ and the two-qubit Molmer-S{\o}rens{\o}n gate  $XX(\theta)\equiv e^{-i\theta X_i\otimes X_j}$ as native gates of the quantum computer. Here, $X$, $Y$, and $Z$ are Pauli matrices.

\subsection{Local depolarizing noise}
\label{app:noise_depol}
The local depolarizing noise channel with the error rate $p$ is 
\begin{equation}
    \mathcal D_p: \rho\mapsto (1-p)\rho+ \frac{p}{3} ( X^\dagger\rho X+ Y^\dagger\rho Y+ Z^\dagger\rho Z).
    \label{eq:depol2}
\end{equation}
This is consistent with the definition in~\ref{eq:depol} and is included here for the reader's convenience. We choose $p=2.425\times10^{-3}$ for simulations of the clean and dirty setup and $p=2.425\cdot10^{-3} f$ where $f=1/n,2/n,\dots,(n-1)/n$ for the simulations with all dirty qubits and variable error rates.  The error channels  are applied to the dirty qubits after each layer of non-parallelizable gates as shown in  Figure~\ref{fig:depol_noise}.   
\begin{figure}[t]
    \centering
    \begin{quantikz}
    \lstick[wires=2]{Noisy}     & \gate[style={},wires=2]{XX(\theta_i)} & \gate[wires=1,style={rounded corners}]{\mathcal D_p}   & \gate[style={dashed}, wires=1]{XX(\theta_i)}\gategroup[4,steps=1,style={dashed,rounded corners, inner xsep=2pt},background]{{PBC}}&\gate[wires=1,style={rounded corners}]{\mathcal D_p}   & \gate[wires=1]{R_Z(\theta_{i+1})} \qw &\gate[wires=1,style={rounded corners}]{\mathcal D_p} & \qw    \\
                                    & \qw                                              & \gate[wires=1,style={rounded corners}]{\mathcal D_p}   & \gate[style={}, wires=2]{XX(\theta_i)}                                                                                 &\gate[wires=1,style={rounded corners}]{\mathcal D_p}   & \gate[wires=1]{R_Z(\theta_{i+1})} \qw &\gate[wires=1,style={rounded corners}]{\mathcal D_p} & \qw    \\
    \lstick[wires=2]{Clean}         & \gate[style={},wires=2]{XX(\theta_i)}& \qw                            & \qw                                                                                                                                                        &\qw                                                    & \gate[wires=1]{R_Z(\theta_{i+1})} \qw &\qw& \qw    \\
                                    & \qw                                              & \qw                            & \gate[style={dashed}, wires=1]{XX(\theta_i)}                                                                                                          &\qw                                                    & \gate[wires=1]{R_Z(\theta_{i+1})} \qw &\qw& \qw   
    \end{quantikz}
    
    \caption{
        \textbf{Local depolarizing noise model.} A 4-qubit HVA circuit with the local depolarizing noise and 2 dirty qubits. The depolarizing noise channel $\mathcal D_p$ is applied to all the dirty qubits after each layer of non-parallelizable gates.  
        }\label{fig:depol_noise}
\end{figure}
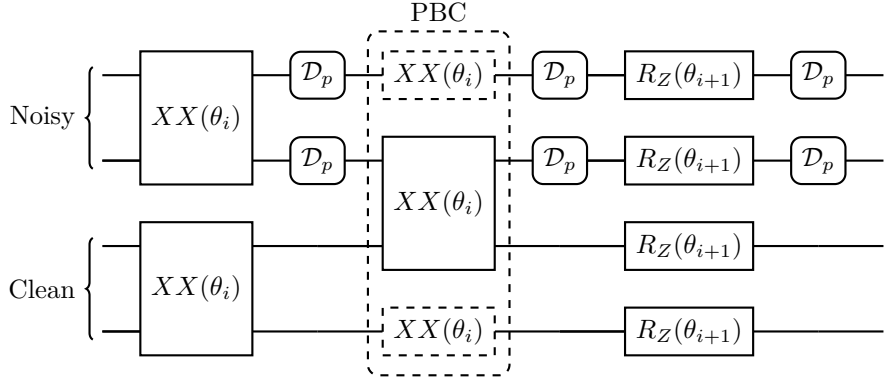

\subsection{Trapped-ion noise model.}
\label{app:noise_real}

The realistic noise model is based on  Trout \textit{et al. }\cite{trout2018simulating}. The noisy quantum  gates $R_X$, $R_Y$, $R_Z$ and $XX$ are defined with noise channels  $\mathcal S_X$, $\mathcal S_Y$, $\mathcal S_Z$ and $\mathcal C$, respectively:
\begin{align}
    \mathcal S_X = &\ \mathcal W_{p_d}\circ\mathcal D_{p_{\textrm{dep}}}\circ\mathcal R_{X}^{p_\alpha} \,,\nonumber \\
      \mathcal S_Y = &\ \mathcal W_{p_d}\circ\mathcal D_{p_{\textrm{dep}}}\circ\mathcal R_{Y}^{p_\alpha}\,,\\ \nonumber
       \mathcal S_Z = &\ \mathcal W_{p_d}\circ\mathcal D_{p_{\textrm{dep}}}\circ\mathcal R_{Z}^{p_\alpha}\,,\\ \nonumber
    \mathcal C = & \left[\mathcal W_1^{p_{d_1}}\otimes \mathcal W_2^{p_{d_2}} \right]
      \circ \left[\mathcal D_1^{p_{\textrm{dep}}}\otimes \mathcal D_2^{p_{\textrm{dep}}} \right]
      \circ\mathcal H_{p_\textrm{xx}}\circ\mathcal H_{p_h}.
\end{align}

Here, $\mathcal R_{X}^{p_\alpha}\rho=(1-p_\alpha)\rho+p_\alpha X \rho X$, $\mathcal R_{Y}^{p_\alpha}\rho=(1-p_\alpha)\rho+p_\alpha Y \rho Y$, $\mathcal R_{Z}^{p_\alpha}\rho=(1-p_\alpha)\rho+p_\alpha Z \rho Z$  are imprecisions of the angle of rotation about the axes $X,Y,Z$, respectively. $\mathcal D_{p_{\textrm{dep}}} $ is the local  depolarizing channel (\ref{eq:depol2}) with the error rate $p_{\textrm{dep}}$, $\mathcal W_{p_d}\rho=(1-p_d)\rho+Z\rho Z$ is the dephasing channel, and $\mathcal H_p=(1-p)\rho+X_1X_2\rho X_1X_2$ is a two qubit channel representing both imprecise rotation when $p=p_{\rm xx}$ and ion heating when $p=p_h$. The indices $1$, $2$ refer to two qubits at which acts $XX$. Furthermore, for an idling dirty qubit, we apply to it  $\mathcal D(p_{\textrm{idle}}) $.  In the computations, the error rates were chosen as
\begin{align}
p_d&= 1.5\cdot10^{-4} f, \quad   p_{\textrm{dep}}=  8.\cdot10^{-4} f, \nonumber  \\
p_{d_1}&=p_{d_2}=  7.5\cdot10^{-4} f, \quad   p_\alpha= 1.\cdot10^{-4} f ,  \\
p_{\rm xx}&= 1.\cdot10^{-3} f, \quad            p_h=  1.25 \cdot10^{-3} f, \nonumber 
\end{align}
with $f=1$ for the clean and dirty setup  and $f=1/n,2/n,\dots,(n-1)/n$ for all dirty qubits with the variable error rates. 

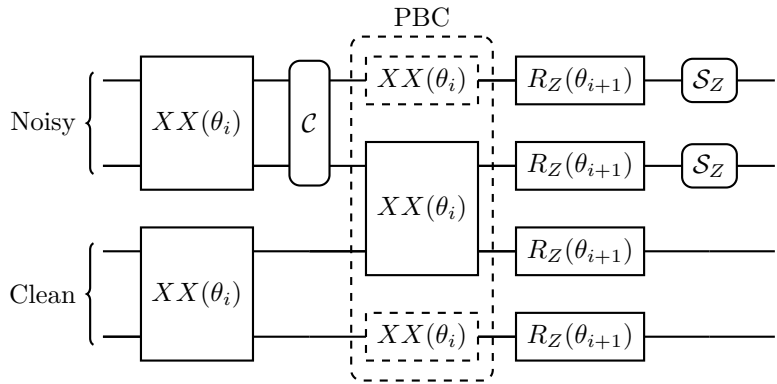
\begin{figure}[ht]
    \centering
    \begin{quantikz}
    \lstick[wires=2]{Noisy}     & \gate[style={},wires=2]{XX(\theta_i)} & \gate[wires=2,style={rounded corners}]{\mathcal C }  & \gate[style={dashed}, wires=1]{XX(\theta_i)}\gategroup[4,steps=1,style={dashed,rounded corners, inner xsep=2pt},background]{{PBC}}& \gate[wires=1]{R_Z(\theta_{i+1})} \qw &\gate[wires=1,style={rounded corners}]{\mathcal S_Z} & \qw    \\
                                    & \qw                                              & \qw   & \gate[style={}, wires=2]{XX(\theta_i)}                                                                                     & \gate[wires=1]{R_Z(\theta_{i+1})} \qw &\gate[wires=1,style={rounded corners}]{\mathcal S_Z} & \qw    \\
    \lstick[wires=2]{Clean}         & \gate[style={},wires=2]{XX(\theta_i)}& \qw                            & \qw                                                                                                                                    & \gate[wires=1]{R_Z(\theta_{i+1})} \qw &\qw& \qw    \\
                                    & \qw                                              & \qw                            & \gate[style={dashed}, wires=1]{XX(\theta_i)}                                                                                      & \gate[wires=1]{R_Z(\theta_{i+1})} \qw &\qw& \qw   
    \end{quantikz}
    
    \caption{
        \textbf{Trapped-ion, realistic noise model.} A 4-qubit HVA circuit with realistic, trapped-ion  noise and 2 dirty qubits.  When a two-qubit gate is applied between two dirty qubits, it is followed by a 2-qubit noise channel $\mathcal{C}$. All the $R_Z$ gates  acting at the dirty qubits are followed by their noise channels $\mathcal{S}_Z$.  If a two-qubit $XX$ gate acts in between a dirty and a clean qubit we assume that it is noiseless. 
        }\label{fig:realistic_noise}
\end{figure}

In the case of a $XX$ gate acting at a dirty and a clean qubit, we assume that it is a noiseless gate. In practice modeling an XX gate in between an error-corrected qubit and a dirty qubit  requires taking into account details of the  error correction code This is outside of the scope of this work, so we assume no error for such a gate. This assumption allows us   to upper bound performance of a quantum computer  with quantum error-corrected  qubits. Figure~\ref{fig:realistic_noise} shows how the noise channels are applied for our HVA simulations.

In the case of our HVA simulations, the choice of perfect XX gates in between dirty and clean qubits  causes  deviations from a  scaling of the gradient magnitude with the $n_d/n$ ratio. For  our choice of the  dirty qubits forming a contiguous block of qubits, the  number of the  noisy XX gates is $L(n_d-1)$ for $1 \le n_d \le n-1$ and $Ln_d$ for $n_d = 0,n$. Therefore,  a fraction of the noisy two-qubit gates in the circuit scales as $(n_d-1)/n$  for $1 \le n_d \le n-1$ causing the gradient magnitude to scale with $(n_d-1)/n$ rather than $n_d/n$  as can be seen in Figure~\ref{fig:RealScaled}.

\section{Additional numerical results} 
\label{app:additionalResults}

Here we show results for the gradient behavior in the case of  4 and 6-qubit HVA circuits simulated with  the realistic and depolarizing noise models. Figure~\ref{fig:numericalResultappendix} shows the mean absolute value of the derivative of the cost  function  with respect to a parameter $\overline{|\partial_{\theta_k} C|}$  averaged over both $28$ instances of random HVA  parameters and the HVA parameters.    plotted versus depth of the ansatz  for the depolarizing and realistic noise models. The behavior is very similar to that of the 8-qubit HVA shown in Figure~\ref{fig:numericalResultLayer}.

Figure~\ref{fig:numericalResultappendix2} displays   $\overline{|\partial_{\theta_k} C|}$ versus the total error  rate for the 4 and 6-qubit HVA circuits.  We find very good quality of the collapse similar to  the 8-qubit case shown in Figure~\ref{fig:numericalResultScaled}. We also  observe that for  the depolarizing noise  small deviations from the perfect collapse  are growing with increasing $n$. We leave further investigation of this phenomenon to future work.  

Figure~\ref{fig:low_depth} shows the behavior of $\overline{|\partial_{\theta_k} C|}$ for lower layer numbers $L=1-30$, $n=4,6,8,10$, and the realistic noise model. Again for each triple of $n,n_d,L$ values we use $28$ random HVA instances to estimate $\overline{|\partial_{\theta_k} C|}$.  Apart from the very low $L\le n$, the behavior of the derivative depends primarily on the $n_d/n$ ratio displaying a decay with an increasing $n_d$ similar to the larger $L$ case.  A different behavior observed for very small $L$ is not surprising as for such shallow circuits one can expect dependence of the impact of the errors to be strongly localized due to the causal cone effects, unlike for the deeper circuits. Furthermore, we see in that regime strong dependence of the results on $L$ that overshadows subtler $n_d$ effects. For $n_d >0$ and larger $L$,  the decay of the gradient with $L$ is approximately exponential, though this behavior is less clear than for larger $L$ due to seemingly random fluctuations superimposed on the dominant trend, which might be finite sample effects.

\begin{figure}
    \centering
    \subfloat[Subfigure 1][Depolarizing noise model - 4 Qubits]{
    \includegraphics[width=0.48\textwidth]{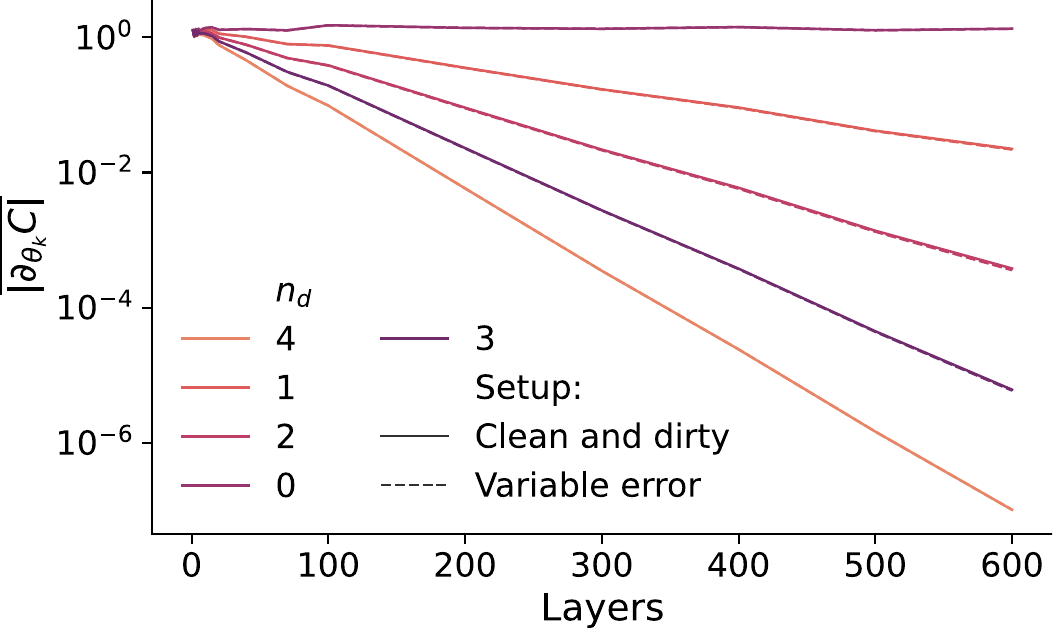}
    \label{fig:DepolLayer4}}
    \subfloat[Subfigure 2][Realistic noise model - 4 Qubits]{
    \includegraphics[width=0.48\textwidth]{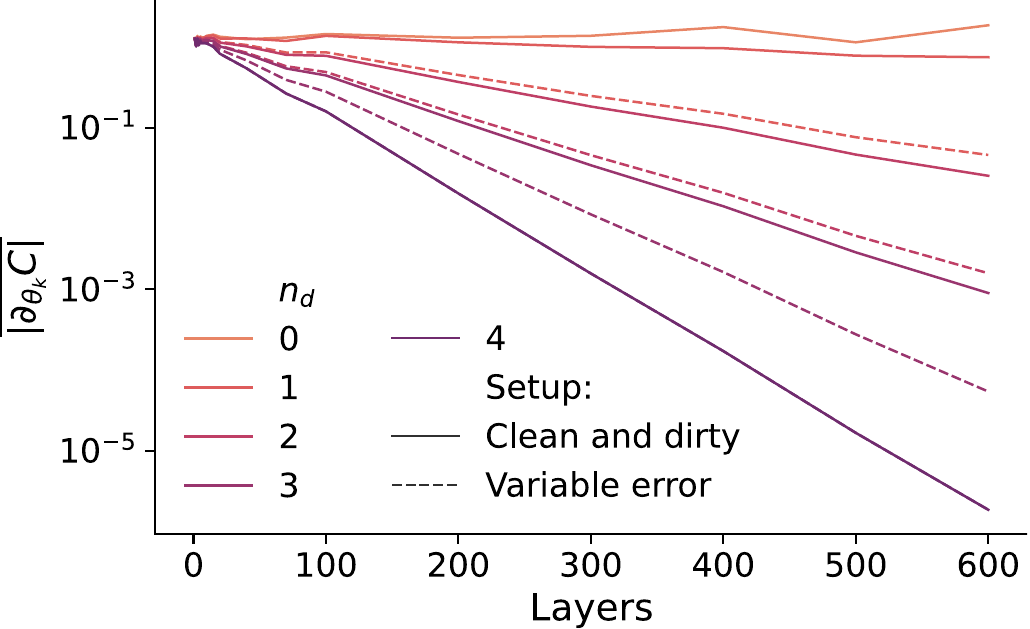}\label{fig:RealLayer4}}\hspace{0mm}
    \subfloat[Subfigure 3][Depolarizing noise model - 6 Qubits]{
    \includegraphics[width=0.48\textwidth]{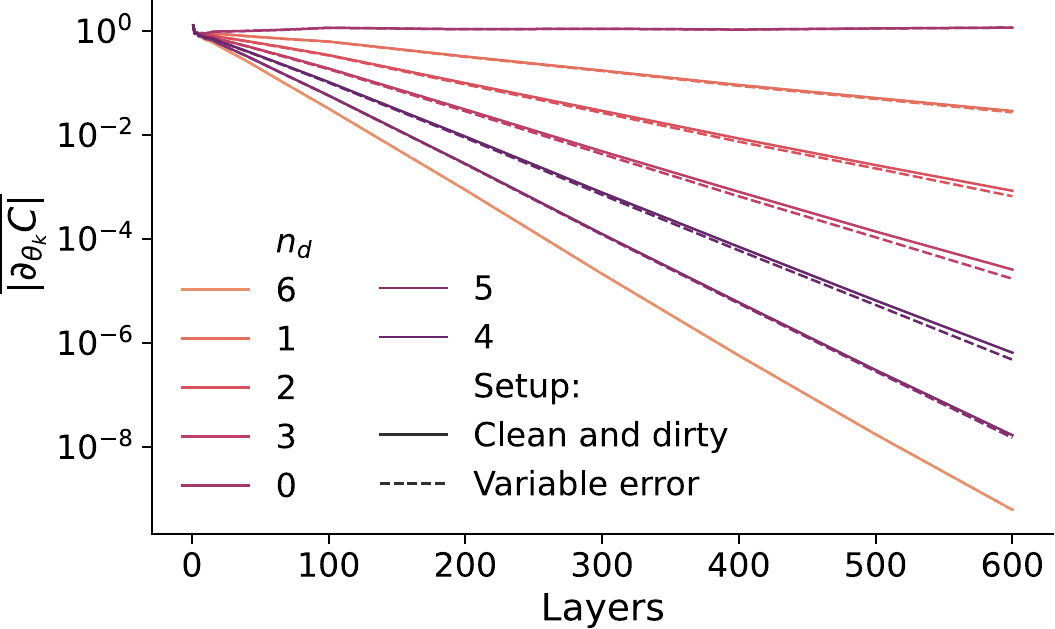}
    \label{fig:DepolLayer6}}
    \subfloat[Subfigure 4][Realistic noise model - 6 Qubits]{
    \includegraphics[width=0.48\textwidth]{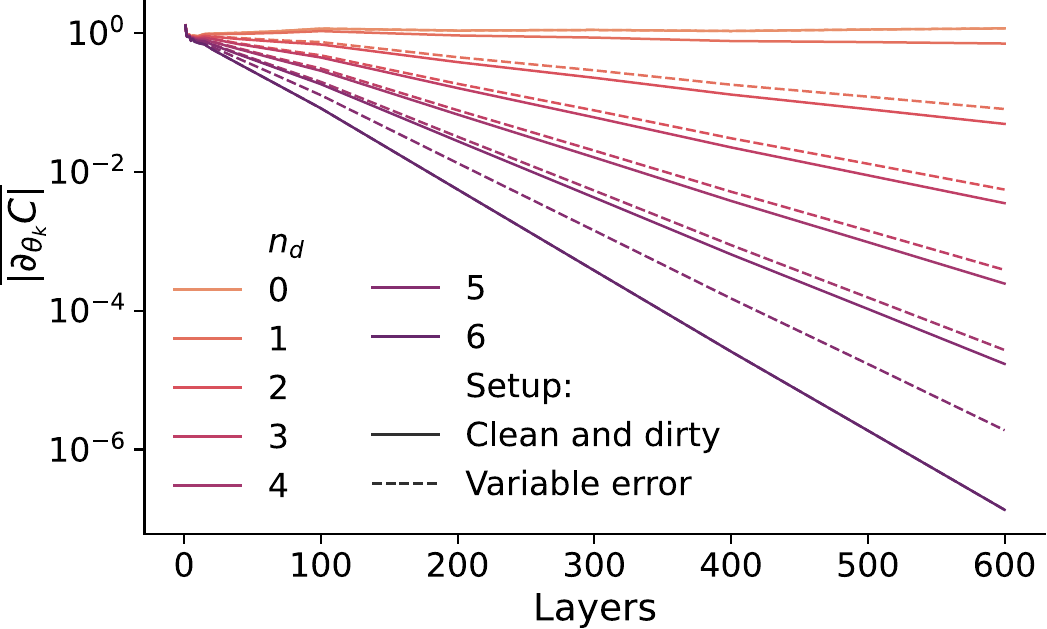}\label{fig:RealLayer6}}
    \caption{\textbf{HVA gradient scaling with respect to the ansatz depth.} Here we show the  mean absolute value of the partial derivative of the cost  function  with respect to a parameter $\overline{|\partial_{\theta_k} C|}$   for   4 (a,b) and 6-qubit (c,d)  systems.  We plot the results for both the clean and dirty and the  variable noise setups  in the same way as in Figure~\ref{fig:numericalResultLayer}. As for the 8-qubit case we find that $\overline{|\partial_{\theta_k} C|}$ decays exponentially with increasing $L$. Furthermore, we observe  that $\overline{|\partial_{\theta_k} C|}$ decays  exponentially with increasing $n_d$, similarly as for the  8-qubit HVA circuits. 
    }\label{fig:numericalResultappendix}
\end{figure}

\begin{figure}
    \centering
    \subfloat[Subfigure 1][Depolarizing noise model - 4 Qubits]{
    \includegraphics[width=0.48\textwidth]{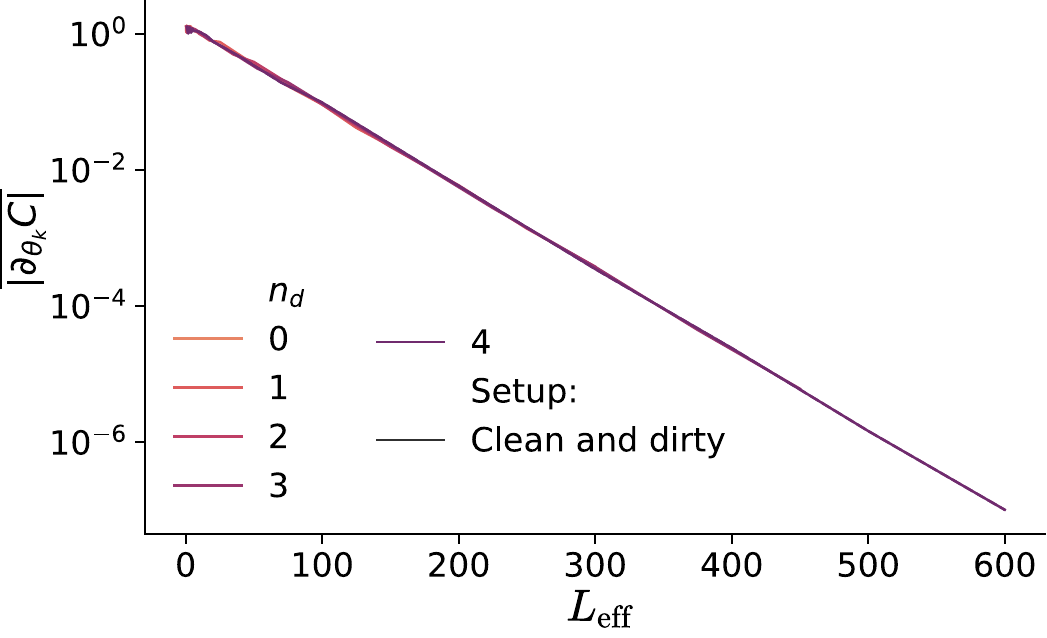}
    \label{fig:ScaledDepolLayer4}}
    \subfloat[Subfigure 2][Realistic noise model - 4 Qubits]{
    \includegraphics[width=0.48\textwidth]{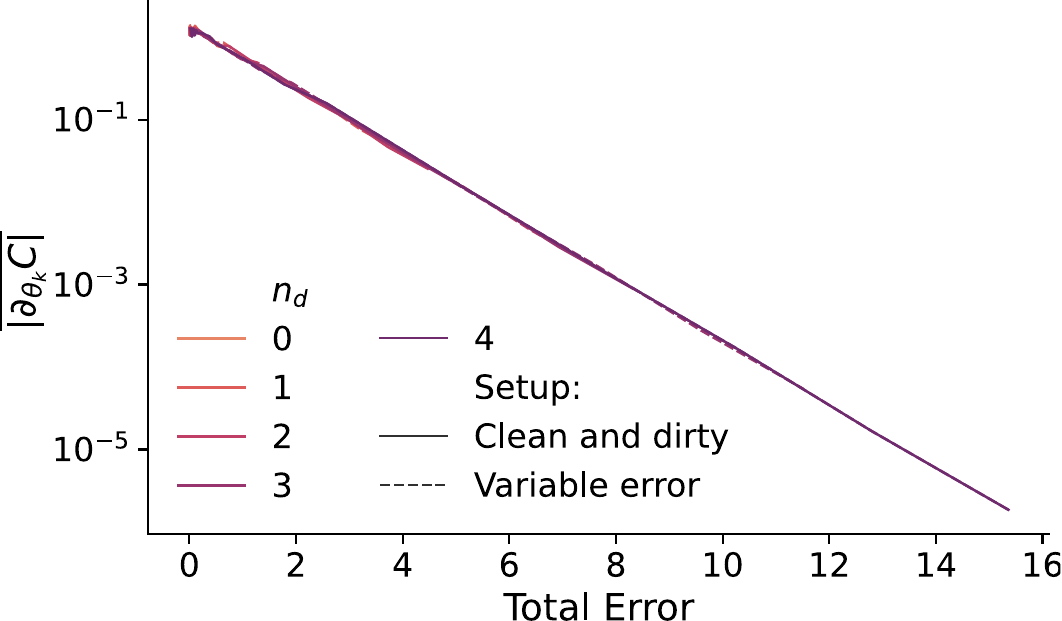}\label{fig:ScaledRealLayer4}}\hspace{0mm}
    \subfloat[Subfigure 3][Depolarizing noise model - 6 Qubits]{
    \includegraphics[width=0.48\textwidth]{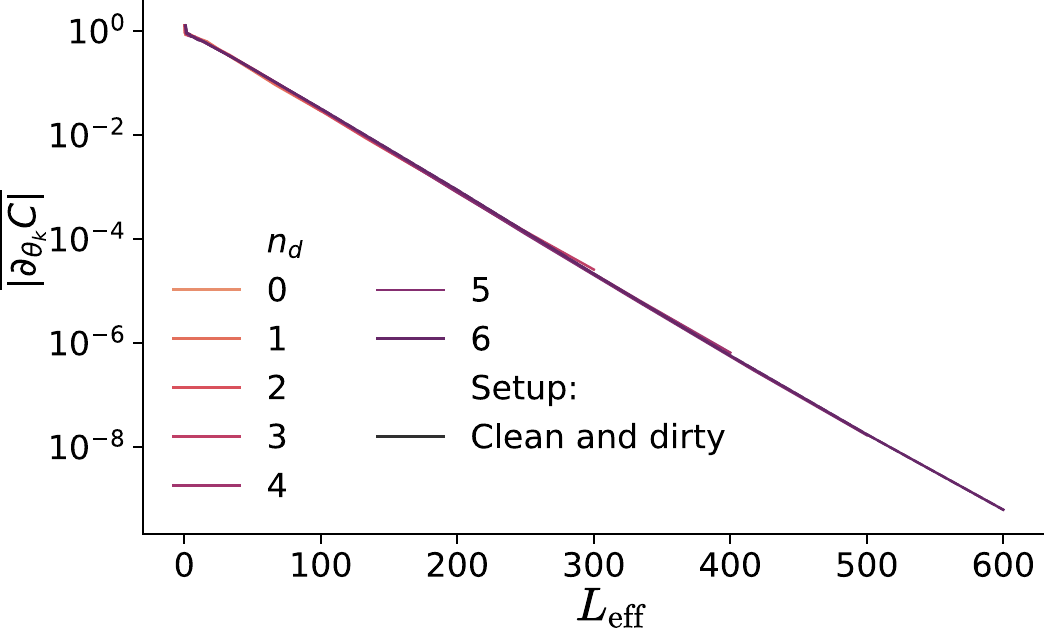}
    \label{fig:ScaledDepolLayer6}}
    \subfloat[Subfigure 4][Realistic noise model - 6 Qubits]{
    \includegraphics[width=0.48\textwidth]{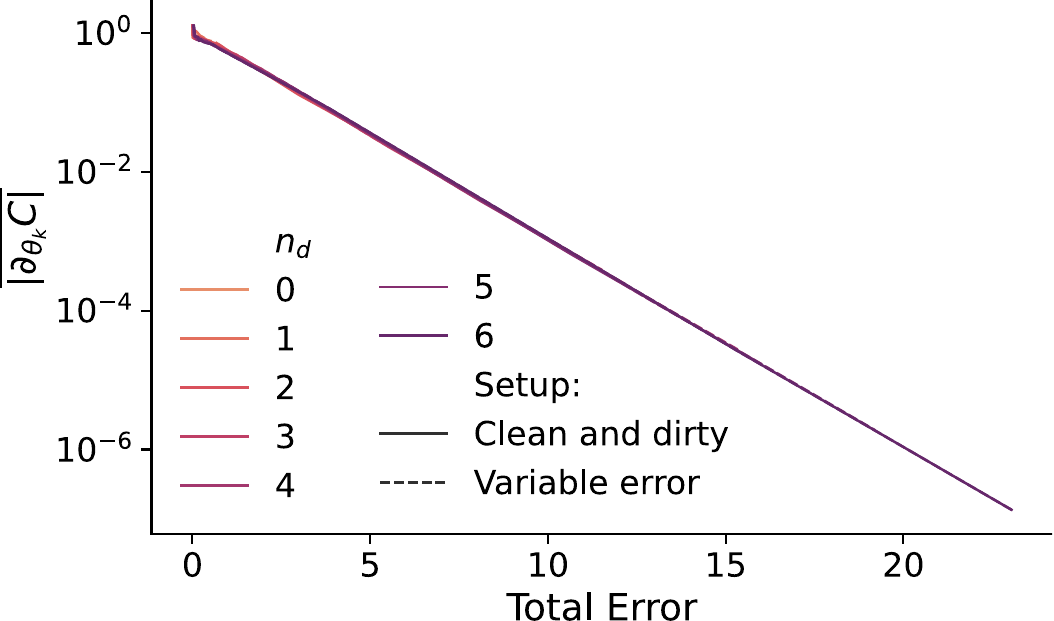}\label{fig:ScaledRealLayer6}}
    \caption{\textbf{HVA $\overline{|\partial_{\theta_k} C|}$ scaling with respect to  the total error rate.} We plot here the  mean absolute value of the partial derivative of the cost  function from Figure~\ref{fig:numericalResultappendix}  versus the total error rate. For  the local depolarizing noise  the total error rate is proportional to   $L_{\textrm{eff}}=L\frac{n_d}{n}$. As for the 8-qubit case (see Figure~\ref{fig:numericalResultScaled}) we observe very good collapse quality.
    }\label{fig:numericalResultappendix2}
\end{figure}

\begin{figure}
    \centering
    \subfloat[Subfigure 1][Realistic noise model - 4 Qubits]{
    \includegraphics[width=0.48\textwidth]{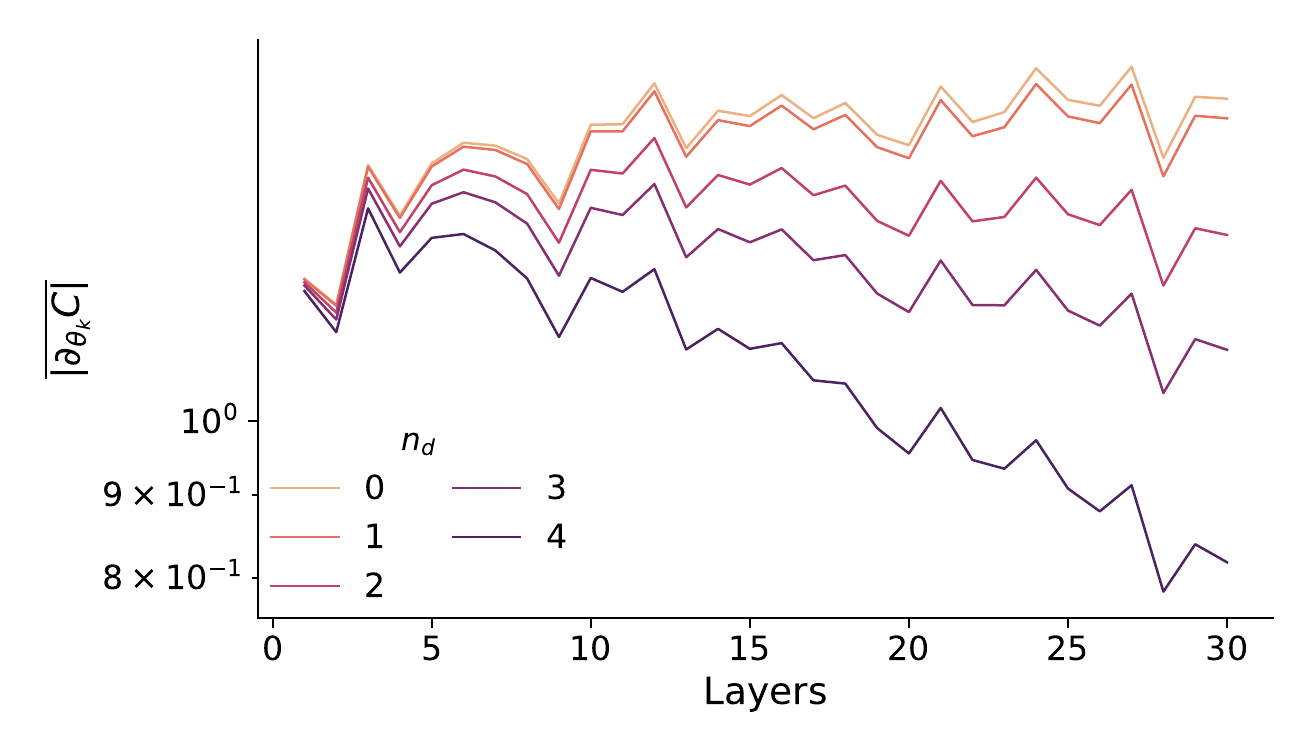}
    \label{fig:low4}}
    \subfloat[Subfigure 2][Realistic noise model - 6 Qubits]{
    \includegraphics[width=0.48\textwidth]{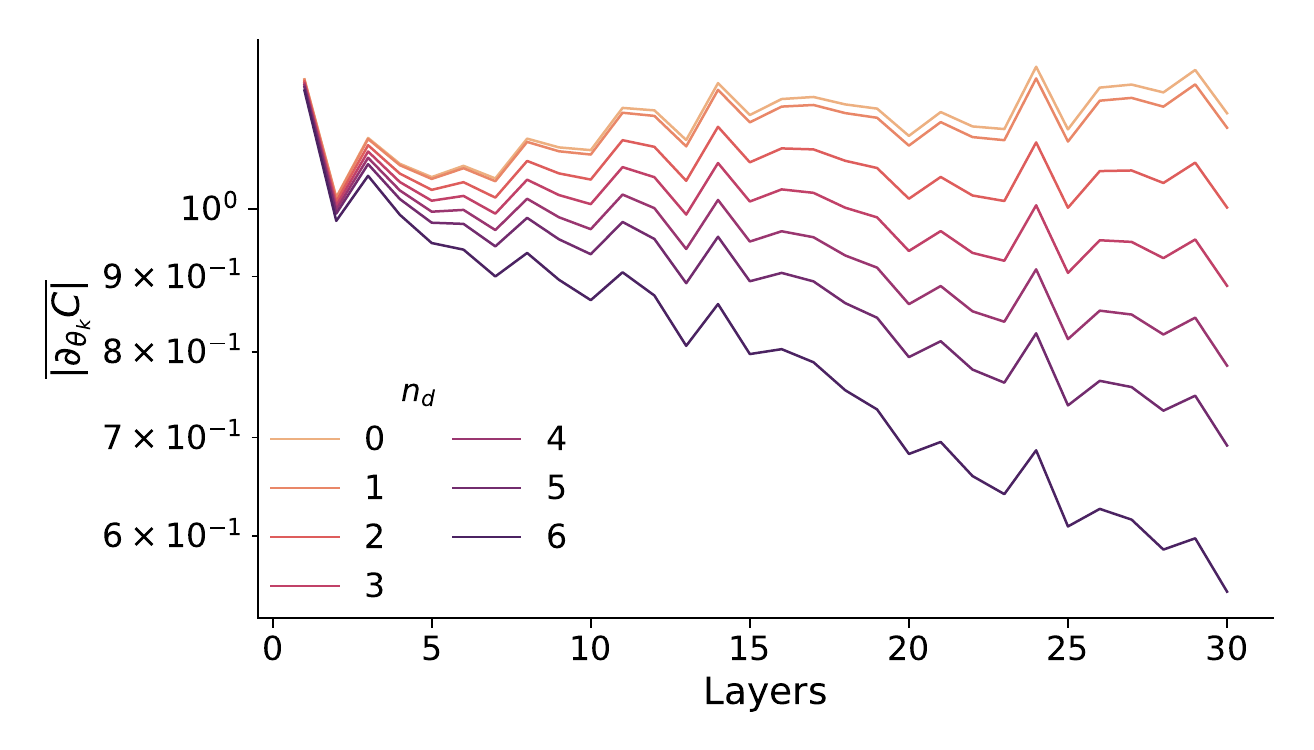}\label{fig:low6}}\hspace{0mm}
    \subfloat[Subfigure 3][Realistic noise model - 8 Qubits]{
    \includegraphics[width=0.48\textwidth]{figures/final_review/6.pdf}
    \label{fig:low8}}
    \subfloat[Subfigure 4][Realistic noise model - 10 Qubits]{
    \includegraphics[width=0.48\textwidth]{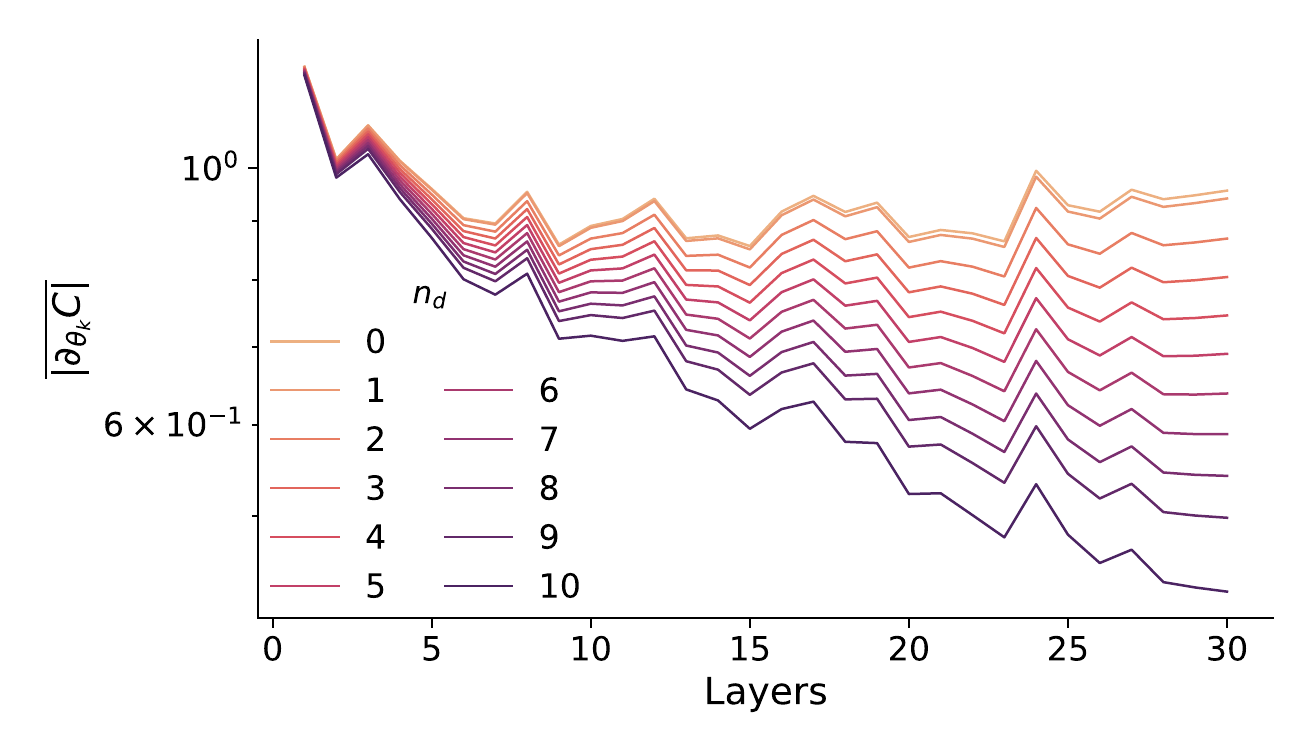}\label{fig:low10}}
    \caption{\textbf{HVA $\overline{|\partial_{\theta_k} C|}$ for 4,6 8, and 10 qubits at low layer depths.} Here we have plotted the realistic noise model simulations for qubit counts 4 in (a), 6 in (b), 8 in (c), and 10 in (d) at depths of up to 30 layers, computed at each layer count. We see that the general pattern seen in the main text holds, with higher noisy qubit counts decaying faster than lower ones. 
    }\label{fig:low_depth}
\end{figure}

\section{Repository}\label{app:repository}
We have prepared a repository which can be found at \href{https://github.com/danielbultrini/Clean\_Dirty\_Qubits}{https://github.com/danielbultrini/Clean\_Dirty\_Qubits}. This contains data frames with numerical results and functions required to plot all the figures shown. In addition, to reproduce the depolarizing noise model results, an open-source Qiskit-based code is provided. This code is not optimized but does implement the depolarizing clean and dirty model. Its aim is to make it easier for interested readers to expand or alter it for their own purposes. 

\end{document}